\documentclass[12pt]{amsart}

\usepackage{geometry}
\usepackage[comma,numbers]{natbib}
\usepackage[inline]{enumitem}
\usepackage{graphicx}
\usepackage[textfont=it,tableposition=above,font=small]{caption}	%\usepackage{framed}			% framed
\usepackage[foot]{amsaddr}	% for title affiliation etc
\usepackage{url}			% for url command in cite{Lopez_Romo2009_working}
\usepackage{xargs}			% newcommandx
\usepackage{amssymb}		% subsetneq (subset but not equal)
\usepackage[ruled,vlined]{algorithm2e}
\usepackage{stmaryrd}       % double left/right bracket

\usepackage{silence}
\usepackage{todonotes}
\usepackage{pgffor}

\RequirePackage[colorlinks,citecolor=blue,urlcolor=blue]{hyperref}
\usepackage[norefs,nocites]{refcheck}
\usepackage{multirow}						% multirows in tables
\usepackage{arydshln}						% dashed line in tables
\usepackage{psfrag}
\usepackage[symbol]{footmisc}
\usepackage{mathtools}						% \coloneqq (:=)

\usepackage{color}
\usepackage{xcolor}
\usepackage{tikz-cd}
\usepackage{pifont}

\newtheorem{theorem}{Theorem}
\newtheorem{lemma}{Lemma}

\theoremstyle{definition}
\newtheorem*{example}{Example}
\newcommand{\HD}{hD}                        % Tukey depth
\newcommand{\R}{\mathbb R}					% reals
\DeclareMathOperator{\dist}{d}				% distance of a point and a set
\newcommand{\Sph}[1][d-1]{\mathbb{S}^{#1}}	% unit sphere
\newcommand{\half}{\mathcal H}

\DeclareMathOperator{\convOp}{conv}			          		% convex hull
\newcommand{\conv}[1]{\convOp\left(#1\right)}
\newcommand{\aff}[1]{\llbracket #1 \rrbracket}           % affine hull
\DeclareMathOperator{\intrOp}{int}			          		% interior
\newcommand{\intr}[1]{\intrOp\left(#1\right)}

\newcommand{\aA}{{\upshape{\textbf{(A)}}}}
\newcommand{\aB}{{\upshape{\textbf{(B)}}}}
\newcommand{\aC}{{\upshape{\textbf{(A$^2$)}}}}
\newcommand{\aD}{{\upshape{\textbf{(A$^3$)}}}}
\newcommand{\aCmb}{{\upshape{\textbf{(C)}}}}

\newcommand{\pkg}[1]{{\normalfont\fontseries{b}\selectfont #1}}
\let\proglang=\textsf
\let\code=\texttt

\newcommand{\bigO}[1]{\mathcal O\left( #1 \right)}
\DeclarePairedDelimiter\floor{\lfloor}{\rfloor}

\title{On exact computation of Tukey depth central regions}

\author{V\'it Fojt\'ik$^{1}$}
\author{Petra Laketa$^{1}$}
\author{Pavlo Mozharovskyi$^{2}$}
\author{Stanislav Nagy$^{1}$}

\email{nagy@karlin.mff.cuni.cz}

\address{\hspace{-1em}$^1$
	Charles University, Prague,
	Faculty of Mathematics and Physics,
	Czech Republic}
	
\address{$^2$ LTCI, 
	Telecom Paris, Institut Polytechnique de Paris,
	France}

\date{\today}

\begin{document}

\begin{abstract}
The Tukey (or halfspace) depth extends nonparametric methods toward multivariate data. The multivariate analogues of the quantiles are the central regions of the Tukey depth, defined as sets of points in the $d$-dimensional space whose Tukey depth exceeds given thresholds $k$. We address the problem of fast and exact computation of those central regions. First, we analyse an efficient Algorithm~\aA{} from \citet{Liu_etal2019}, and prove that it yields exact results in dimension $d=2$, or for a low threshold $k$ in arbitrary dimension. We provide examples where Algorithm~\aA{} fails to recover the exact Tukey depth region for $d>2$, and propose a modification that is guaranteed to be exact. We express the problem of computing the exact central region in its dual formulation, and use that viewpoint to demonstrate that further substantial improvements to our algorithm are unlikely. An efficient \proglang{C++} implementation of our exact algorithm is freely available in the \proglang{R} package \pkg{TukeyRegion}.
\end{abstract}



\maketitle

%\spacingset{1.5} % DON'T change the spacing!

%\maketitle

%
%
%

\section{Introduction: Tukey depth and its central regions}

The \emph{Tukey depth} (or \emph{halfspace depth}, or simply \emph{depth}) is a prominent method of nonparametric analysis of multivariate data. Proposed in 1975 by \citet{Tukey1975}, it is firmly established in nonparametric and robust statistics since the 1990s \cite{Donoho_Gasko1992}. For a point $x \in \R^d$ and a dataset $X = \left\{x_1, \dots, x_n\right\} \subset \R^d$, the Tukey depth\footnote{We consider only the depth for datasets, that is the sample depth. For general measures the depth is typically taken scaled into the interval $[0,1]$, which is obtained by dividing our expression for $\HD$ by $n$. For our purposes, the integer-valued version of the depth is more convenient to work with, but this minor difference is without loss of generality.} of $x$ with respect to $X$ is defined as the minimum number of data points in any halfspace that contains $x$ on its boundary 
    \[ \HD(x;X) = \min_{u \in \Sph} \#\left\{ i \in \left\{1, \dots, n\right\} \colon \left\langle x, u \right\rangle \leq \left\langle x_i, u \right\rangle \right\}. \] 
Here, $\Sph$ is the unit sphere in $\R^d$, and $u \in \Sph$ is the inner normal of the halfspace $\left\{ y \in \R^d \colon \left\langle x, u \right\rangle \leq \left\langle y, u \right\rangle \right\}$. The depth assesses the degree of centrality of $x$ with respect to the geometry of the data cloud $X$. The higher the depth is, the more ``centrally positioned" $x$ is within $X$. While immensely successful in applications \cite{Liu_etal1999, Rousseeuw_Ruts1999, Zuo_Serfling2000}, the exact and fast computation of the depth for $d > 2$ has been resolved only relatively recently \cite{Dyckerhoff_Mozharovskyi2016}. 

Perhaps even more important than the depth of a single point $x$ are the \emph{central regions} of the depth of $X$ at levels $k \geq 1$, defined as the upper level sets
    \[  \HD_k(X) = \left\{ x \in \R^d \colon \HD(x;X) \geq k \right\}.   \]
The central regions form a system of nested compact convex polytopes. For $k = 1$ we obtain the convex hull of $X$. The smallest non-empty set $\HD_k(X)$ is a generalisation of the median set to $\R^d$-valued data, and is called the \emph{Tukey} (or \emph{halfspace}) \emph{median} of $X$. In case when a unique point representing the median set is required, the barycentre of the median set is frequently singled out. The central regions of $X$ describe the shape of the dataset. Interestingly, they encode the complete information present in $X$, as there exist methods for reconstructing the data points from the central regions only \cite{Laketa_Nagy2021, Struyf_Rousseeuw1999}. The central regions are vital in many applications --- they are used in data visualisation, anomaly detection, classification, or the construction of multivariate boxplots, to give a few examples.

If the dataset $X$ is in general position,\footnote{A set of points in $\R^d$ is in general position if no $d+1$ of these points lie in a hyperplane.} each polytope $\HD_{k+1}(X)$ lies in the interior of the previous region $\HD_{k}(X)$, for $k\geq 1$ \cite[Lemma~6]{Struyf_Rousseeuw1999}. \emph{Throughout this paper, we assume that the points of $X$ are in general position.} That assumption is standard in the depth literature. It greatly facilitates both the analysis and computation. If the dataset $X$ is sampled from a distribution with a density, it is in general position almost surely.

We are concerned with the exact computation of the regions $\HD_k(X)$. Recently, an efficient algorithm was proposed in the literature \cite[Algorithm~2]{Liu_etal2019}. We refer that program Algorithm~\aA{} for brevity. Algorithm~\aA{} cleverly combines the ideas of projecting the data into two-dimensional subspaces, and a consecutive breadth-first search strategy along those hyperplanes determined by $d$ data points that may determine a piece of the boundary of the region $\HD_k(X)$. It is fast, and possible to be used also for data of dimension $d>2$. Despite not being proved theoretically, ample empirical evidence presented in \cite{Liu_etal2019} suggested that the algorithm may give the exact Tukey depth regions $\HD_k(X)$ for any $d \geq 1$ and $k \geq 1$.

We begin in Section~\ref{sec:A} by analysing the exactness of Algorithm~\aA{} from a theoretical perspective. In Sections~\ref{sec:d2} and~\ref{sec:k2} we prove that Algorithm~\aA{} does indeed give exact results in dimension $d=2$ for any $k \geq 1$, and in any dimension $d > 2$ for $k=1,2$. In Section~\ref{sec:not exact} we proceed with a surprising negative result. We provide a dataset of $n = 12$ points in $\R^3$ in general position where Algorithm~\aA{} fails to recover the Tukey depth central region, meaning that Algorithm~\aA{} is not exact in general.

Based on our observations, in Section~\ref{sec:exact} we modify Algorithm~\aA{}, and prove that our new version called Algorithm~\aB{} recovers the central regions for any dataset of points in general position, for any dimension $d \geq 1$ and any $k \geq 1$. An extensive simulation study presented in Section~\ref{sec:simulation} highlights that despite Algorithm~\aB{} is more complex than Algorithm~\aA{}, in the task of computing multiple central regions, the two procedures are on par in terms of speed. In particular, Algorithm~\aB{} is well suited for the exact computation of the complete collection of central regions of $X$, including the Tukey median set.

In the concluding Section~\ref{sec:dual} we recast our results in view of the so-called dual graph of $X$, a useful tool for visualisation and diagnostics for the central regions. Using dual graphs, we demonstrate that none of the several appealing simplifications of our Algorithm~\aB{} cannot guarantee exactness. It therefore appears unlikely that a procedure substantially simpler than our Algorithm~\aB{} would be able to recover the exact Tukey depth regions. The extensive technical proofs of our main results are gathered in the Appendix.

Our proofs employ notions from convex geometry, and rely heavily on the polarity theory for convex polytopes. After defining the essential notations in Section~\ref{sec:notations}, we therefore begin our exposition by a brief overview of the necessary theory on polar polytopes in Section~\ref{sec:polarity}.

\subsection{Notations}  \label{sec:notations}

The boundary of a set $A \subset \R^d$ is denoted by $\partial A$, and its interior by $\intr{A} = A \setminus (\partial A)$. We write $\conv{A}$ for the convex hull of $A$, and $\aff{A}$ for its affine hull.\footnote{Convex hull of $A$ is defined as the intersection of all convex sets that contain $A$; its affine hull is the intersection of all translations of vector subspaces (that is, affine subspaces of $\R^d$) that contain $A$.} In most situations this notation will be applied to a finite set $A = \left\{ a_1, \dots, a_m \right\}$, where we write also $\conv{a_1, \dots, a_m}$ for the convex hull of these points, and $\aff{a_1, \dots, a_m}$ for their affine hull. For example, for $a_1 \ne a_2 \in \R^d$, $\conv{a_1,a_2}$ and $\aff{a_1,a_2}$ stand for the line segment and the infinite line delimited by $a_1$ and $a_2$, respectively.

Denote by $X$ a set of $n > d+1$ data points in $\R^d$ in general position. A \emph{ridge} is any subset of $d-1$ points from $X$. We say that a hyperplane is \emph{observational} if it is determined by the affine hull of $d$ points from $X$, meaning that it is the unique hyperplane that contains all those $d$ points. A closed halfspace whose boundary is an observational hyperplane is called an \emph{observational halfspace}. An observational halfspace $H$ is \emph{relevant (at level $k$)} if the complementary open halfspace $\R^d \setminus H$ contains exactly $k-1$ points from $X$; we write $\half(k)$ for the set of all relevant halfspaces at level $k$.  A \emph{relevant hyperplane (at level $k$)} is the boundary hyperplane $\partial H$ of a relevant halfspace $H\in\half(k)$. We also say that a relevant halfspace $H \in \half(k)$ (or its boundary $\partial H$) \emph{cuts off} $k-1$ points from $X$.

Two halfspaces $H,H'\in \half(k)$ are \emph{(mutually) reachable} (in $\half(k)$) if \begin{enumerate*}[label=(\roman*)] \item they are \emph{neighbouring}, meaning that their boundaries contain the same ridge, or \item there exists $H''\in\half(k)$ that is reachable in $\half(k)$ from both $H$ and $H'$.\end{enumerate*} Note that this definition is given recursively --- $H$ and $H'$ are reachable if and only if there exists a finite sequence of halfspaces $\left\{ H_j \right\}_{j=1}^J \subseteq \half(k)$ such that $H_1 = H$, $H_J = H'$, and for each $j = 1, \dots, J-1$ the boundaries of $H_j$ and $H_{j+1}$ share a ridge. Starting from a given ridge $I \subset X$, the search strategy through ridges employed in our algorithms finds all relevant halfspaces $H' \in \half(k)$ that are reachable from (any) halfspace $H \in \half(k)$ such that $I \subset \partial H$. The collection of all these halfspaces reachable from $H \in \half(k)$ (or equivalently reachable from the ridge $I$) will be called an \emph{orbit} of $H$ (or $I$) in $\half(k)$. Mathematically speaking, the orbit of $H\in\half(k)$ is the transitive closure of the binary relation of halfspaces in $\half(k)$ being neighbouring to $H$. The orbits in $\half(k)$ partition $\half(k)$ into equivalence classes.

\subsection*{The rationale of the algorithms} Since $\HD_k(X)$ can be defined as the intersection of all the elements of $\half(k)$ \cite[Proposition~6]{Rousseeuw_Ruts1999}, the problem of finding the central region of $X$ at level $k$ reduces to the task of identifying all relevant halfspaces $\half(k)$. In the sequel, we are therefore concerned with algorithms for finding all halfspaces from $\half(k)$, or equivalently, all $d$-tuples of points from $X$ whose affine hulls cut off exactly $k-1$ data points from $X$.

\subsection{Preliminaries: Polar polytopes} \label{sec:polarity}

We use duality considerations from convex geometry \cite[Section~2.4]{Schneider2014}. First, we recall basic definitions and facts about polar polytopes. A polytope $P \subset \R^d$ is the convex hull of a finite number of points in $\R^d$. In this work we deal only with full-dimensional polytopes, that is polytopes whose interior is non-empty. A face of $P$ is a convex subset $F \subseteq P$ that satisfies that $x, y \in P$ and $(x+y)/2 \in F$ implies $x, y \in F$. The single point faces $F$ of $P$ are called vertices of $P$, the one-dimensional faces are the edges of $P$. A $(d-1)$-dimensional face of $P$ is a facet of $P$. For a polytope $P \subset \R^d$ that contains the origin in its interior, the \emph{polar polytope} of $P$ is defined as \cite[Section~2.1]{Schneider2014}
	\[	P^\circ = \left\{ x \in \R^d \colon \left\langle x, y \right\rangle \leq 1 \mbox{ for all }y \in P \right\}.	\]
Denote by $F_1, \dots, F_m$ all the facets of $P$. The \emph{conjugate face} of $F_j$, $j=1,\dots,m$, is
	\begin{equation}	\label{conjugate face}
	\hat{F}_j = \left\{ x \in P^\circ \colon \left\langle x, y \right\rangle = 1 \mbox{ for all }y \in F_j \right\}.	
	\end{equation}
By \cite[formula~(2.28)]{Schneider2014} we know that each $\hat{F}_j$ is a point in $\R^d$, and \cite[Lemma~2.4.5]{Schneider2014} gives that $P^\circ = \conv{\hat{F}_1, \dots, \hat{F}_m}$. Thus, the vertices $\hat{F}_j$ of $P^\circ$ correspond to the outer normals of the facets $F_j$, $j=1,\dots,m$, in the sense that we can write
	\begin{equation}	\label{P by duality}
	P = \bigcap_{j=1}^m \left\{ x \in \R^d \colon \left\langle x, \hat{F}_j \right\rangle \leq 1 \right\}.	
	\end{equation}
A pair of vertices $\hat{F}_j$ and $\hat{F}_k$ of $P^\circ$ is joined by an edge on the boundary of $P^\circ$ if and only if the facets $F_j$ and $F_k$ share a $(d-2)$-dimensional face $F$ of $P$ \cite[Theorem~2.4.9 and formula~(2.28)]{Schneider2014}, and in particular $F$ is then the convex hull of some $d-1$ vertices of $P$. In what follows, this theory will be applied in the situation when the vertices of $P$ are a subset of the dataset $X$, and both $F_j$ and $F_k$ determine boundaries of halfspaces from $\half(k)$. In that case, in terms of reachability introduced in Section~\ref{sec:notations}, we see that $\hat{F}_j$ and $\hat{F}_k$ are joined by an edge if and only if the affine hulls of $F_j$ and $F_k$ are neighbouring in $\half(k)$.

\section{Theoretical analysis of Algorithm~{\aA}}    \label{sec:A}

For a dataset $X \subset \R^d$ and a level $k \geq 1$, Algorithm~\aA{} for finding the central region $\HD_k(X)$ is based on the following general procedure which we call \textbf{\textsf{RidgeSearch}}. This scheme encompasses a whole family of algorithms in the spirit of both Algorithms~1 and~2 from \cite{Liu_etal2019}. The indicated \textbf{Steps 1--8} refer to the description of Algorithm~2 from \cite{Liu_etal2019}.

%\problem{Make sure that the table float with the algorithm is properly placed on a page}

% change the caption of the algorithm to Main algorithm.
\renewcommand{\algorithmcfname}{Main algorithm \textsf{RidgeSearch}}
\renewcommand{\thealgocf}{}

   \begin{algorithm}
	\caption{The search procedure through ridges in the boundaries of $H \in \half(k)$ from \cite{Liu_etal2019}.}
	\begin{enumerate}[label=\textbf{(S$_{\arabic*}$)}, ref=\upshape{\textbf{(S$_{\arabic*}$)}}]
		\item \label{S1} Initialisation (generalisation of \textbf{Steps~1} and~\textbf{2} from~\cite{Liu_etal2019}):
		\begin{itemize}
			\item[\ding{228}] A queue $\mathcal Q$ of ridges of $X$ is initialised, using an algorithm-specific rule.
			\item[\ding{228}] A set of found relevant halfspaces $\mathcal H_k$ is initialised to be empty.
		\end{itemize} 
		\item \label{S2} The main loop runs through the queue of ridges $\mathcal Q$ (\textbf{Step~3} from~\cite{Liu_etal2019}): 
		
		\textbf{For} each $I \in \mathcal Q$ \textbf{do} (\textbf{Step~4} from~\cite{Liu_etal2019})
		%\For{$I \in \mathcal Q$ (\textbf{Step~4})}{
		% find all $H \in \half(k)$ such that $I \subset \partial H$, and
		\begin{itemize}
			\item[\ding{228}] Find all relevant halfspaces reachable from $I$ in $\half(k)$, i.e. the orbit of $I$ in $\half(k)$. This is performed by a search strategy where first, all halfspaces in $\half(k)$ containing $I$ are found, and then this process is iterated for all ridges of those halfspaces. The search finishes when the complete orbit $O$ of the ridge $I$ in $\half(k)$ is found.
			\item[\ding{228}] Append the whole orbit $O$ of the ridge $I$ to the set $\mathcal H_k$.
		\end{itemize}
		%}
		\item \label{S3} In the final \textbf{Steps~5--8} from~\cite{Liu_etal2019}, the intersection of all halfspaces from $\mathcal H_k$ is found. This is the output of the algorithm, being an estimate of the central region $\HD_k(X)$. It is hoped that in $\mathcal H_k$ we recovered all orbits, or equivalently the whole set $\half(k)$. In that case $\HD_k(X) = \bigcap \mathcal H_k$.
	\end{enumerate}
\end{algorithm}
    
The crucial part of \textbf{\textsf{RidgeSearch}} is the selection of the initial set of ridges $\mathcal Q$ at stage~\ref{S1}. This stage is the one where the particular instances of the procedure \textbf{\textsf{RidgeSearch}} differ. 

Trivially, if all possible $\binom{n}{d-1}$ ridges of the dataset $X$ are included in $\mathcal Q$, \textbf{\textsf{RidgeSearch}} yields an exact solution. This was observed already in \cite{Liu_etal2019}, where the last algorithm was presented as Algorithm~1, and is also called the \emph{combinatorial algorithm}. For brevity, we call this exact program Algorithm~\aCmb{} (for ``combinatorial''). It is currently the only relatively fast procedure for the computation of the central regions with an exactness guarantee. Nevertheless, as argued already in \cite{Liu_etal2019}, the initial selection of all $\binom{n}{d-1}$ ridges in $\mathcal Q$ makes Algorithm~\aCmb{} slow for many setups. 

In contrast, in the original fast Algorithm~\aA{} from \cite{Liu_etal2019}, the initial set of ridges $\mathcal Q$ at stage~\ref{S1} is chosen according to the following heuristic (\textbf{Step~2} in \cite{Liu_etal2019}):
    \begin{enumerate}[label=\textbf{(A$_{\arabic*}$)},label=\upshape{\textbf{(A$_{\arabic*}$)}}]
        \item \label{A1} A single ridge $I$ of $d-1$ points on the boundary of the convex hull of $X$ is found. 
        \item \label{A2} All\footnote{At this step we slightly simplify Algorithm~2 from \cite{Liu_etal2019}. In the original version, only two relevant halfspaces $H \in \half(k)$ of this type are found in \textbf{Step~2(d)} \cite[p.~686]{Liu_etal2019}. Of course, our inclusion of (possibly) more than two relevant halfspaces in~\ref{A2} makes Algorithm~\aA{} to search through more ridges. Thus, if Algorithm~2 from \cite{Liu_etal2019} is exact, then so must be our Algorithm~\aA{}. This difference is of no importance for our exposition, and does not alter any of our conclusions.} relevant halfspaces $H \in \half(k)$ that contain $I$ in their boundary hyperplanes are obtained, and all ridges determined by points of $X$ in those hyperplanes are placed into $\mathcal Q$. 
        \item \label{A3} Finally, all ridges formed by $d-2$ points of $I$ and a single point cut off by any halfspace $H \in \half(k)$ from step~\ref{A2} are added to $\mathcal Q$.
    \end{enumerate}
Algorithm~\aA{} therefore involves steps~\ref{S1}--\ref{S3} of \textbf{\textsf{RidgeSearch}}, with the initialisation~\ref{A1}--\ref{A3} at step~\ref{S1}. For a more detailed description of Algorithm~\aA{} we refer to \cite{Liu_etal2019}. Here we provide only a small motivating example for $d=2$ that is summarized in Figure~\ref{figure:motivation}. In this example we see that for any $k \geq 1$, Algorithm~\aA{} gives an exact central region. We will show below that this is not a coincidence, and for $d=2$ Algorithm~\aA{} is always exact. Another example of Algorithm~\aA{} will be given for $d=3$ in Section~\ref{sec:not exact}.

\begin{figure}
    \centering
    \includegraphics[width=.35\textwidth]{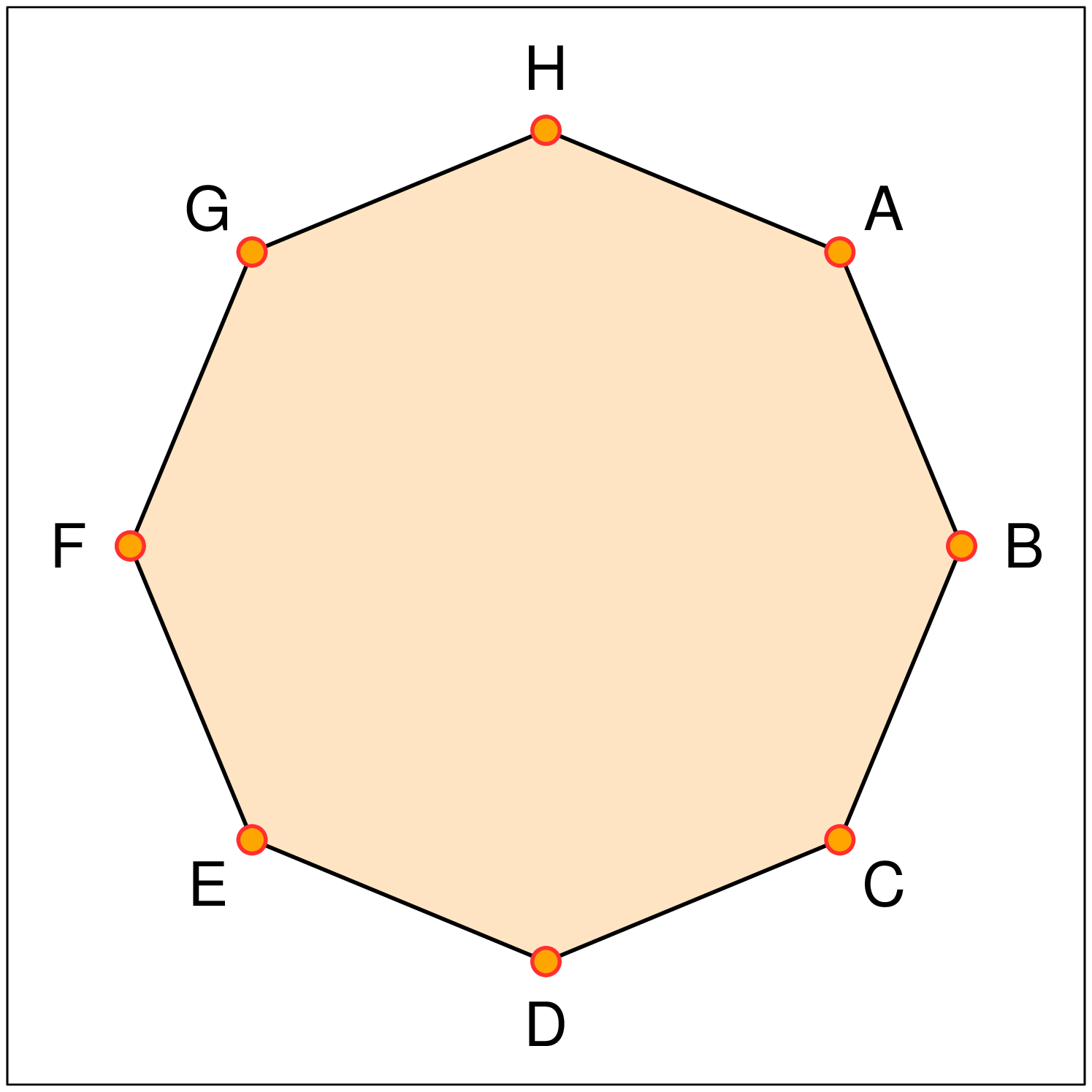} \quad \includegraphics[width=.35\textwidth]{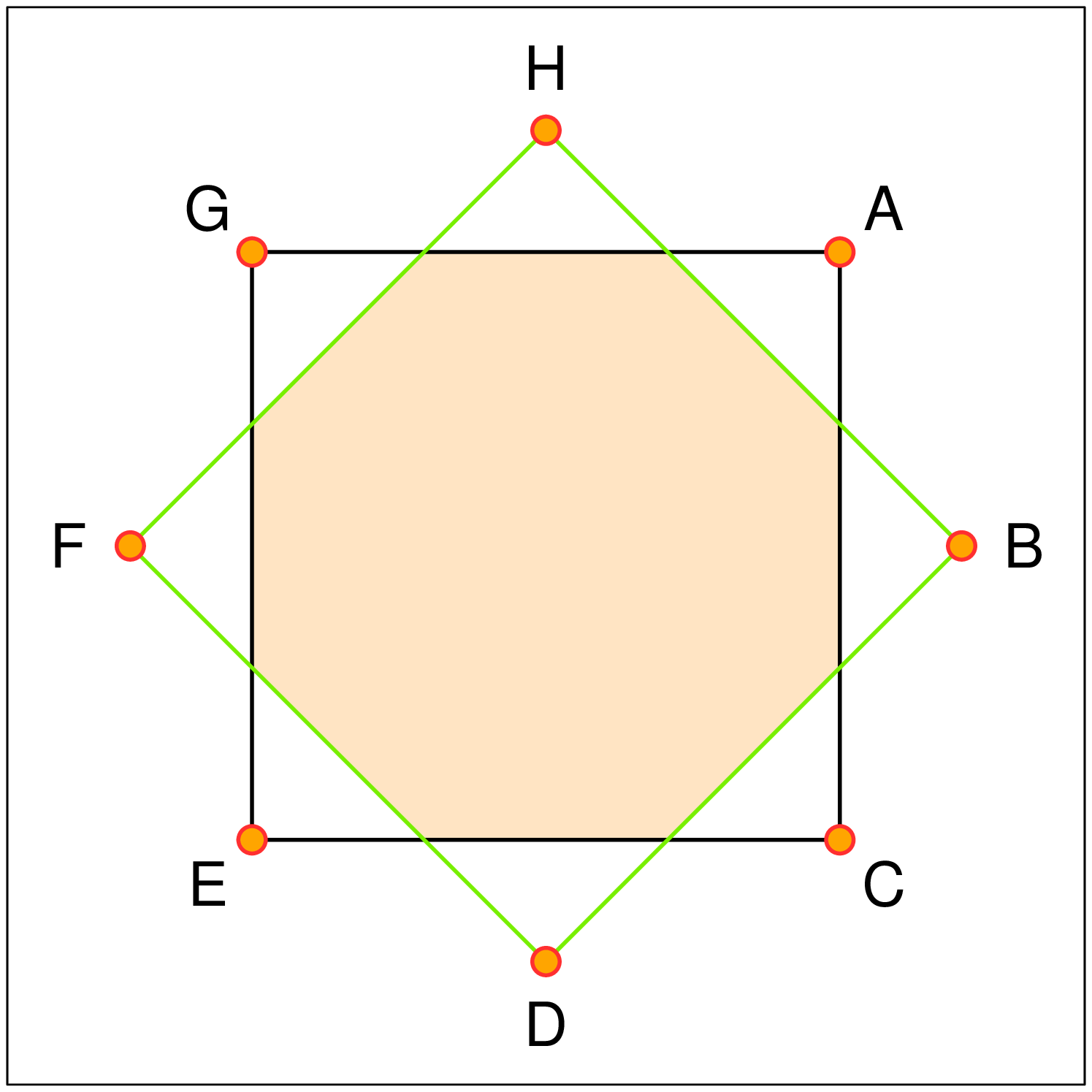} \\[1em]
    \includegraphics[width=.35\textwidth]{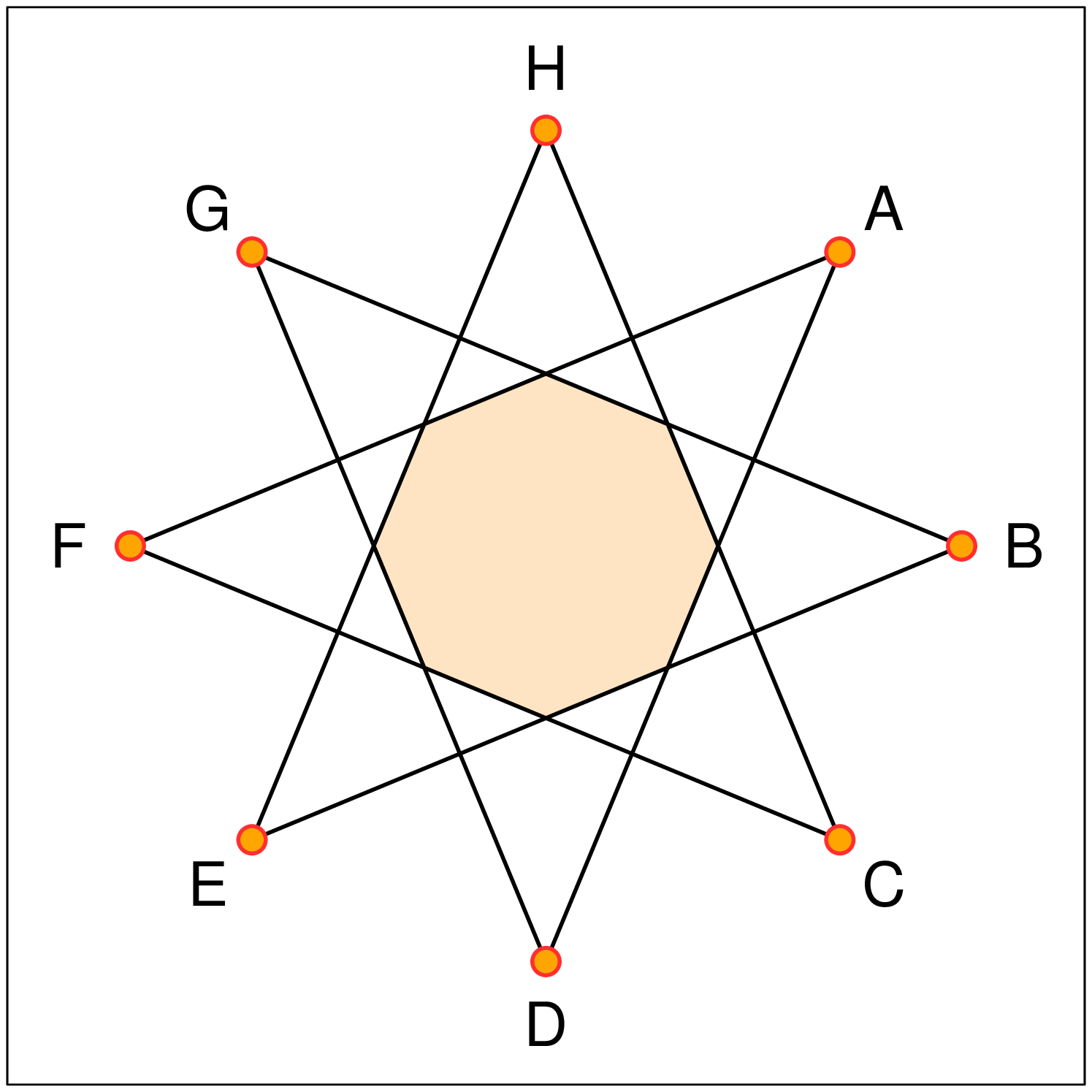} \quad \includegraphics[width=.35\textwidth]{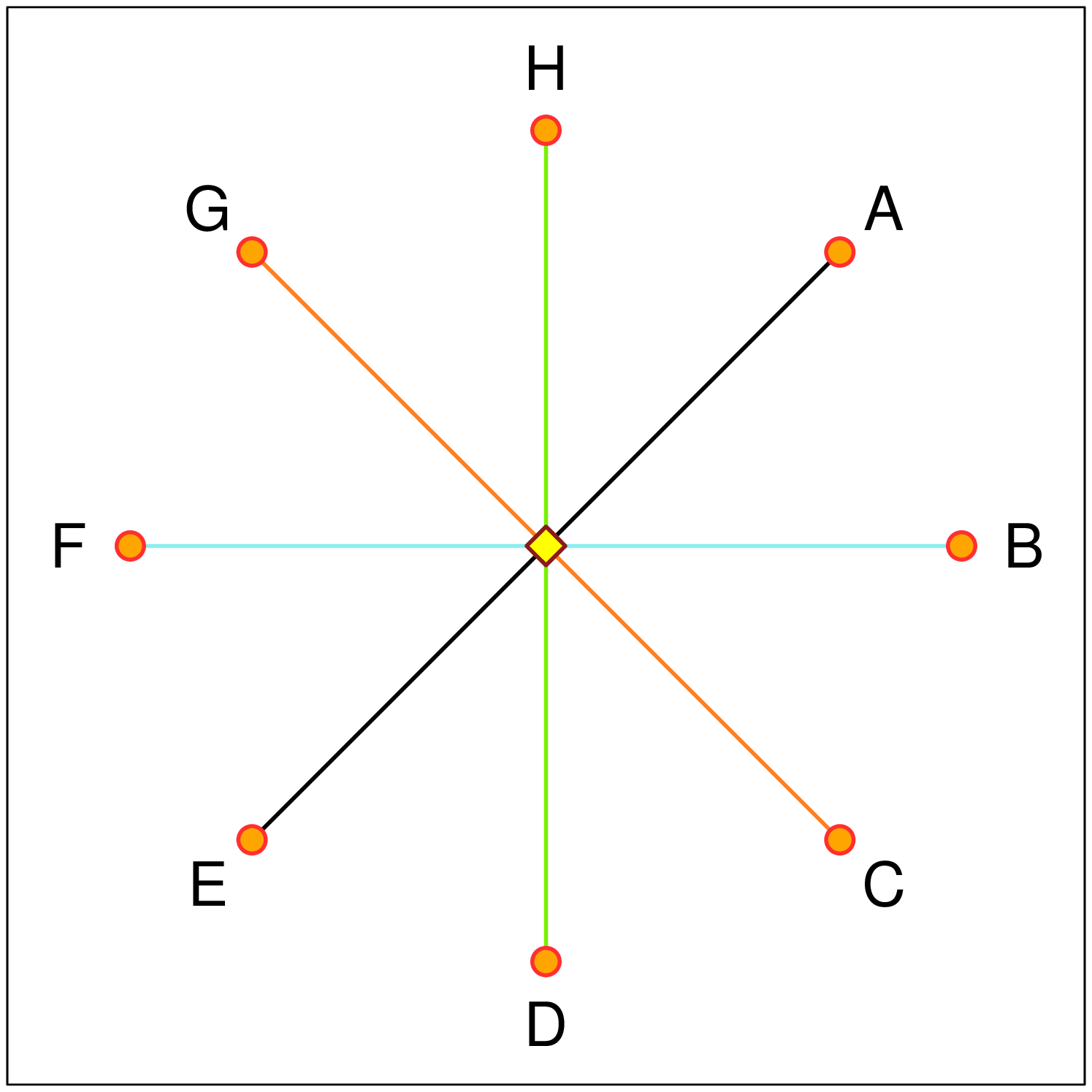}    
    \caption{A motivating example: A dataset $X$ of $n = 8$ points $A$--$H$ in $\R^2$, the relevant lines forming boundaries of halfplanes from $\half(k)$ (thick lines), and the central regions (shaded regions) for $k=1$ (top left), $k=2$ (top right), $k=3$ (bottom left), and $k=4$ (bottom right). For $k=1$ and $k=3$ all ridges (that is, data points for $d=2$) lie in a single orbit, and \textbf{\textsf{RidgeSearch}} is exact with any single initial ridge in $\mathcal Q$. For $k=2$ there are two orbits --- starting from the ridge $A$, the only two relevant halfplanes from $\half(2)$ that contain $A$ in their boundary are those given by $\aff{A,C}$ and $\aff{A,G}$. The search initialised at $A$ in~\ref{S2} recovers its orbit $\left\{A,C,E,G\right\}$. A second orbit in $\half(2)$ is given by $\left\{B,D,F,H\right\}$. Using Algorithm~\aA{} we initialise at a ridge (say) $A$ in~\ref{A1}, then include in $\mathcal Q$ also ridges $C$ and $G$ in~\ref{A2}, and finally take into $\mathcal Q$ also ridges $B$ and $H$ in~\ref{A3}. Thus, Algorithm~\aA{} gives an exact result. For $k=4$ we have four orbits in $\half(4)$, given by $\left\{A,E\right\}$, $\left\{B,F\right\}$, $\left\{C,G\right\}$, and $\left\{D,H\right\}$, respectively. The median set $\HD_4(X)$ is the single point (yellow diamond) in the centre of the figure.}
    \label{figure:motivation}
\end{figure}

In \cite[Section~4.1]{Liu_etal2019}, an extensive simulation study was performed to demonstrate that in tens of thousands of simulated runs, Algorithm~\aA{} always recovered the exact central region. In what follows we validate some of those positive results from a theoretical perspective. Afterwards, we construct an example showing that Algorithm~\aA{} may fail for $d>2$.

\subsection{Algorithm~\aA{} is exact for \texorpdfstring{$d=2$}{d=2}}    \label{sec:d2}

Assume for a moment that $d = 2$ and $1 \leq k < n/2$.\footnote{The extreme cases $k \geq n/2$ are not interesting, because clearly $\HD_k(X) = \emptyset$ if $k > n/2$ \cite[see e.g.][Theorem~1]{Liu_etal2020}. Furthermore, for $n$ even and $k = n/2$, if the set $\HD_{n/2}(X)$ is non-empty, then $X$ is a halfspace symmetric \cite{Zuo_Serfling2000c} configuration of points. By \cite[Proposition~1]{Liu_etal2020} for $d>2$ this is impossible. For $d=2$ and the non-trivial case $n > 2$ this is possible only for $\HD_{n/2}(X)$ a single point set \cite[Theorem~3.1]{Zuo_Serfling2000c}, a situation which is not covered by \textbf{\textsf{RidgeSearch}}. In fact, it can be shown that for $X$ sampled from an absolutely continuous distribution in dimension $d = 2$, $\HD_{n/2}(X)$ is either empty or a sample point from $X$, with probability one \cite{Pokorny_etal2021}.} Recall that for $x \neq y \in \R^2$ we write $\aff{x, y}$ for the unique line passing through $x$ and $y$. In two dimensions, a halfspace is a halfplane, hyperplanes are lines (determined by two points) and a ridge is just one point. 

\begin{theorem}
	\label{thm:orbitspassthrough}
Let $H \in \half(k)$ be a relevant halfplane cutting off points $U \subseteq X$ from $X$. Then for all orbits $O$ of halfplanes from $\half(k)$, either $H \in O$, or there exists $H^\prime \in O$ and $x_l \in U$ such that $L^\prime = \partial H^\prime$ passes through $x_l$. That is, every orbit in $\R^2$ contains $H$, or a halfplane whose boundary passes through a point cut off by $H$.
\end{theorem}

The detailed proof of Theorem~\ref{thm:orbitspassthrough} is given in Section~\ref{sec:orbits} in the Appendix. As a direct consequence, we obtain our first main result.

\begin{theorem}
\label{theorem:d2}
Algorithm~{\aA{}} in dimension $d=2$ finds all relevant hyperplanes at level $1 \leq k < n/2$ in any dataset $X$ in general position of size $n$. In other words, with the set of initial ridges chosen using the heuristic~\ref{A1}--\ref{A3}, we have $\mathcal H_k = \half(k)$. In particular, Algorithm~{\aA{}} is exact for $d=2$.
\end{theorem}

\begin{proof}
Algorithm~\aA{} starts with at least one relevant hyperplane (line) $L$ in~\ref{A2} and all relevant hyperplanes that pass through points that are cut off by $L$ in~\ref{A3}. Then, it generates the orbits of these hyperplanes. By Theorem~\ref{thm:orbitspassthrough}, this includes all orbits in $\half(k)$.
\end{proof}

As a consequence of the proof of Theorem~\ref{theorem:d2} we can reduce the initial set of ridges $\mathcal Q$ in Algorithm~\aA{} to only $k$ points (ridges) for $d=2$, without losing exactness. For any relevant halfplane $H \in \half(k)$, take the $k-1$ points cut off by $H$ and one of the two points on the boundary of $H$ as the initial set of ridges. Then, Theorem \ref{thm:orbitspassthrough} shows that the procedure \textbf{\textsf{RidgeSearch}} finds all relevant halfspaces and hence gives the exact central region $\HD_k(X)$. 

%%%\problem{[PAVLO: Do you think that it's worth programming this last simplification to TukeyRegion for $d=2$, and remark here that it's implemented?]}

%%%\pavlo{Not sure this would make it faster. On the other hand, there are also other points to improve in the code. Maybe no, for time saving purposes.}

%%%The exact Algorithm~\aA{} should be compared with another exact procedure for the computation of central regions for $d=2$ called the \code{isodepth} \cite{Ruts_Rousseeuw1996}, available in the \proglang{R} package \pkg{depth} \cite{R_depth} (please note, that the two algorithms have the same complexity). \problem{[PAVLO: Can we say that the two algorithms have the same complexity?]} \pavlo{Yes.} A small simulation performed in Section~\ref{sec:simulation} of this paper demonstrates the superiority of Algorithm~\aA{} in all setups. \problem{Update this conclusion according to the results of the study.}

%
%
%
%
%

\subsection{Algorithm~\aA{} is exact for \texorpdfstring{$k=1,2$}{k=1,2} for any \texorpdfstring{$d$}{d}}    \label{sec:k2}

Denote by $C$ be the convex hull of $X$. In our first lemma we deal with the simple case of $k = 1$. In that situation $\HD_1(X) = C$, and all relevant halfspaces in $\half(1)$ form a single orbit as we saw also in Figure~\ref{figure:motivation}. In particular, Algorithm~\aA{} is exact for $k=1$ with any single initial ridge in the queue $\mathcal Q$. 

\begin{lemma}\label{lm:connected}
Any two halfspaces in $\half(1)$ determined by facets of $C$ are mutually reachable in $\half(1)$.
\end{lemma}

The proof of Lemma~\ref{lm:connected} is given in Section~\ref{section:connected} in the Appendix. We now turn our attention to the more interesting case $k=2$. In a series of auxiliary lemmas stated and proved in Section~\ref{section:k2} in the Appendix, we obtain a proof of exactness of Algorithm~\aA{} for $k=2$.

\begin{theorem}\label{theorem:k2}
Algorithm~{\aA{}} finds all relevant halfspaces for $k = 1, 2$ and any dimension $d = 1, 2, \dots$. Consequently, Algorithm~{\aA{}} gives an exact solution for $k=1,2$.
\end{theorem}

Again as a consequence of our proof of Theorem~\ref{theorem:k2} we are able to reduce the initial set of ridges $\mathcal Q$ in Algorithm~\aA{} for $k=2$ to $d$, still keeping the algorithm exact. Indeed, take a single facet of $C$ and consider any collection of $d$ halfspaces from $\half(2)$, each of these halfspaces cutting off one of the $d$ vertices of $C$. The set of all ridges in the boundaries of these halfspaces is placed in the initial queue $\mathcal Q$. Then the proof of Theorem~\ref{theorem:k2} guarantees that for $k=2$ the procedure \textbf{\textsf{RidgeSearch}} gives the exact central region $\HD_2(X)$.

%%%\problem{[PAVLO: Do you think that it's worth programming this last simplification to TukeyRegion for $k=2$, and remark here that it's implemented?]}

%%%\pavlo{Not sure it's worth it...}
%
%
%
%
%

\subsection{Algorithm~\aA{} is not exact in general} \label{sec:not exact}

The output of Algorithm~\aA{} is a  collection of relevant halfspaces $\half_k \subset \half(k)$ at level $k$. Their intersection is therefore always a superset of the corresponding central region $\HD_k(X)$. In the following example we demonstrate that for $k>2$ and $d>2$, Algorithm~\aA{} does not always recover the central region exactly.

\begin{example}
Consider a dataset $X$ of $n = 12$ points in $\R^3$ in general position. Each of these points is labeled by a colour: red, blue or green. We start with red points positioned in the vertices of a regular tetrahedron 
    \[
    r_i = \frac{1}{\sqrt{2}}\left(e_i-\frac{1}{2}\left(1,1,1\right)\right) \mbox{ for }i=1,2,3,\quad \mbox{ and } \quad
    r_4 = \frac{1}{\sqrt{8}}\left(1,1,1\right),
    \]
for $e_1=(1,0,0)$, $e_2=(0,1,0)$ and $e_3=(0,0,1)$. The blue and the green points are respectively placed at $b_i = -0.3\,r_i$ and $g_i = 0.15\,r_i$, for $i=1,\dots,4$. To satisfy the condition of the points being in general position, we rotate a bit the vertices of the blue and the green tetrahedrons, each in a slightly different way. The data is constructed so that the convex hull of $X$ is formed only by the four red points, and so that the green vertices lie outside the blue tetrahedron. An example of such a configuration of points is in Figure~\ref{fig:12points} and in the supplementary \proglang{Mathematica} notebook, where our whole construction is visualised in an interactive display.

\begin{figure}[htpb]
    \centering
    \includegraphics[width=.485\textwidth]{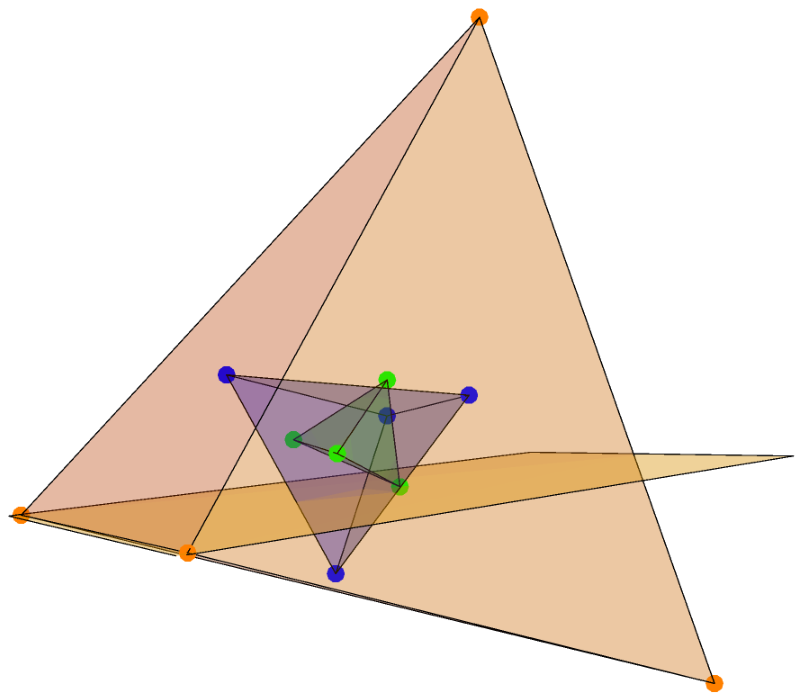}
    \includegraphics[width=.485\textwidth]{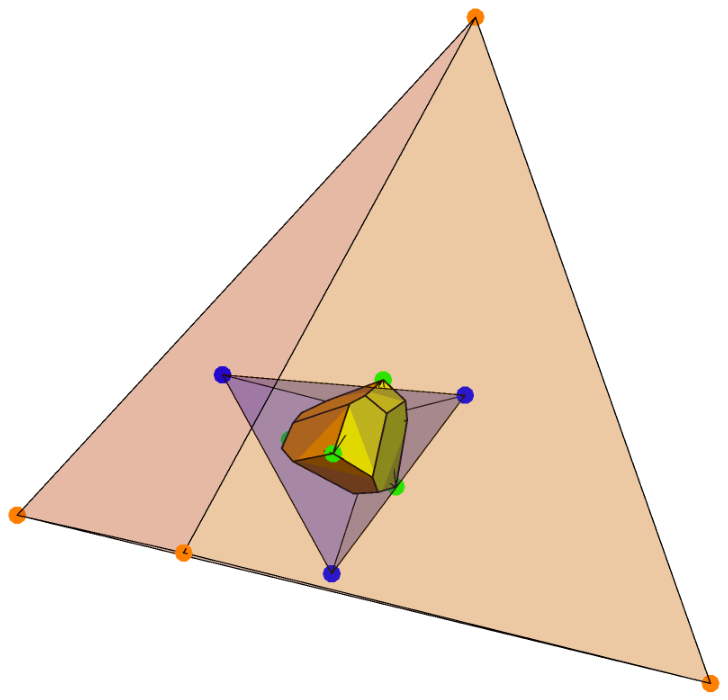} 
    \caption{The point configuration $X$ from Section~\ref{sec:not exact} with a single relevant plane for $k=3$ (left hand panel) and the central region for $k=3$ as computed by Algorithm~{\aA{}} (yellow polytope in the right hand panel). With any choice of an initial ridge and the corresponding relevant planes, no relevant plane formed by three blue points can be reached by Algorithm~{\aA{}}. The green vertices $g_i$, $i=1,\dots,4$ are therefore not cut off by any found relevant plane, and each $g_i$ is declared to lie in the central region $\HD_3(X)$. Each $g_i$ has, however, Tukey depth only $\HD(g_i;X) = 2$. This can be seen also in the right hand panel --- points $g_i$ all lie in the region computed by Algorithm~{\aA{}}.}
    \label{fig:12points}
\end{figure}

Consider the depth level $k=3$. Algorithm~\aA{} begins in step~\ref{A1} with a ridge $E$, i.e. an edge of the convex hull of $X$, which must consist of two red points. The algorithm then proceeds in step~\ref{A2} by finding two relevant planes $P_1$ and $P_2$ that contain $E$. Because of how the points of $X$ are positioned, each of these planes is determined by the two red points incident to $E$, and an additional single green point, see also the left hand panel of Figure~\ref{fig:12points}. Each of these planes cuts off one red and one blue point from $X$. In the initial step we add into the queue $\mathcal Q$ all the ridges obtained by substituting one of the points from the initial ridge $E$ with one point cut off by either $P_1$ or $P_2$ from $X$ in step~\ref{A3}. This gives us the initial set of ridges $\mathcal Q$ whose incident points are of the following colours: R-R, R-B and R-G, where R,G and B represent red, green and blue. 

Note that any plane determined by three blue points is relevant for $k=3$ --- it cuts off exactly two points from $X$, one red and one green. At the same time, no other plane determined by two blue points and one point of another colour is relevant for $k=3$. This is very easy to see in the interactive \proglang{Mathematica} visualisation provided in the Supplementary Material. Therefore, from our initial set of ridges $\mathcal Q$ it is impossible to obtain any ridge coloured as B-B, meaning that Algorithm~\aA{} fails to find any relevant plane determined by three blue points. For that reason, the resulting depth region for $k=3$ obtained by Algorithm~\aA{} contains all the green points $g_i$, $i=1,\dots,4$, which certainly have Tukey depth only $\HD(g_i;X) = 2 < k$. The last claim is seen by considering a halfspace passing through a green point whose boundary plane is parallel to the adjacent blue facet. That halfspace contains only a single red point in its interior. As we intended to show, Algorithm~\aA{} fails to recover the exact central region $\HD_3(X)$. 
\end{example}

%
%
% Exact algorithm for computation
%
%

\section{Algorithm~\aB{}: Exact computation of central regions} \label{sec:exact}

We now turn to the problem of exact computation of the central regions by means of the ridge-wise search strategy \textbf{\textsf{RidgeSearch}}. Our intention is to find an initial set of ridges $\mathcal Q$ guaranteeing the exactness of the procedure. As demonstrated in our example from Section~\ref{sec:not exact} and corroborated in Section~\ref{sec:dual} below, it turns out that in the task of computing $\HD_k(X)$ directly, it is unlikely that an initial set of much less than all $\binom{n}{d-1}$ ridges (as for Algorithm~\aCmb{}) suffices for an exact result. 

We therefore approach the problem in a different way, and argue that \textbf{\textsf{RidgeSearch}} is feasible to be run recursively. We show that given the set of all relevant halfspaces $\half(k-1)$, to obtain all relevant halfspaces $\half(k)$ using \textbf{\textsf{RidgeSearch}} it is enough to initialise the queue $\mathcal Q$ in~\ref{S1} in the following way (see also third bullet point at the end of Section~2 in~\cite{Liu_etal2019}):
    \begin{enumerate}[label=\textbf{(B$_{\arabic*}$)},label=\upshape{\textbf{(B$_{\arabic*}$)}}]
        \item \label{B1} All ridges of points in the boundaries of relevant halfspaces from $\half(k-1)$ are placed into $\mathcal Q$.
    \end{enumerate}
The complete Algorithm~\aB{} involves running steps~\ref{S1}--\ref{S3} with the queue $\mathcal Q$ in step~\ref{S1} chosen using~\ref{B1}; for a summary of our three procedures see Table~\ref{table:summary}. 

\begin{table}
\begin{tabular}{c|c|c|c}
    \textbf{\textsf{RidgeSearch}} version & Algorithm~\aA{} & Algorithm~\aB{} & Algorithm~\aCmb{} \\ \hline
    Initialisation in \textbf{Step~\ref{S1}} & \ref{A1}--\ref{A3} & \ref{B1} & all ridges
\end{tabular}
\caption{A summary of the initialisation used in our three variants of \textbf{\textsf{RidgeSearch}}.}
\label{table:summary}
\end{table}

The initialisation~\ref{B1} typically results in a larger set of initial ridges $\mathcal Q$ than what is considered in~\ref{A1}--\ref{A3}. In Section~\ref{sec:dual} and the examples in the Supplementary Material we however argue that this appears to be needed to guarantee exactness. On the other hand, the set of ridges from~\ref{B1} has usually much less elements than $\binom{n}{d-1}$  needed for Algorithm~\aCmb{}. It turns out that numerically, running our algorithm with initialisation~\ref{B1} several times to compute all central regions at levels $k = 1, \dots, K$ for $K \geq 1$ given is surprisingly not slower than running the fast (and not exact) Algorithm~\aA{} to compute $\HD_k(X)$ for all $k=1,\dots,K$. All this will be demonstrated in numerical studies in Section~\ref{sec:simulation}.

The following theorem is the main ingredient of our exact Algorithm~\aB{} for the computation of the central regions. We describe, given the set of all relevant halfspaces $\half(k)$ at level $k \geq 2$, a way to find all relevant halfspaces $\half(k+1)$ at level $k+1$. We argue that any $H \in \half(k+1)$ can be reached from a halfspace $H' \in \half(k+1)$ that shares a common ridge with a relevant halfspace $\widetilde{H}$ from $\half(k)$. Schematically, we obtain
    \[  \begin{tikzcd}[column sep=huge]
H \in \half(k+1) \arrow[<->]{r}{reachable} & H' \in \half(k+1) \arrow[<->]{r}{neighbouring} & \widetilde{H} \in \half(k).
\end{tikzcd}    \]
A detailed proof of Theorem~\ref{next level} is found in Section~\ref{sec:next level} in the Appendix.

\begin{theorem}   \label{next level}
Let $H\in\half(k+1)$. Then there exists \begin{enumerate*}[label=(\roman*)] \item a halfspace $H'\in\half(k+1)$ that is reachable from $H$ and \item a halfspace $\widetilde{H} \in \half(k)$ that shares a ridge in the boundary with $H'$.\end{enumerate*} In particular, Algorithm~{\aB} defined by~\textbf{\textsf{RidgeSearch}} with the initialisation~\ref{B1} is exact.
\end{theorem}

% {\color{blue}
% \begin{corollary}
% Let $H\in\half(k)$. Then there exists \begin{enumerate*}[label=(\roman*)] \item a halfspace $H'\in\half(k)$ that is reachable from $H$ and \item a halfspace $\widetilde{H} \in \half(k+1)$ that shares a ridge in the boundary with $H'$.\end{enumerate*}
% \end{corollary}
% \begin{proof}
% Denote $H_1=\R^d\setminus \intr{H}$. Because $H_1$ cuts off $k-1$ points from $X$ and $\partial H$ is determined by $d$ data points, we have that $H_1\in\half(n-k-d+1)$. Then we can apply Theorem~\ref{next level} to $H_1$ to obtain $H'_1\in\half(n-k-d+1)$ that is reachable from $H_1$ and $\widetilde{H}_1 \in \half(n-k-d)$ that shares a ridge in the boundary with $H'_1$. Take $H'=\R^d\setminus \intr{H'_1}$ and $\widetilde{H}=\R^d\setminus\intr{\widetilde{H}_1}$. Then $H'\in\half(k)$ and $\tilde{H}\in\half(k+1)$, which proves the result.
% \end{proof}
% }

% Theorem~\ref{next level} alone guarantees the exactness of our procedure~\ref{S1}--\ref{S3} with the initial set~\ref{B1}. In what follows we call this method Algorithm~\aB{}, for brevity.

%
%
%

\section{Empirical comparison: Numerical studies}   \label{sec:simulation}

\subsection{Algorithm~\aA{} and \texorpdfstring{\code{isodepth}}{isodepth} for \texorpdfstring{$d=2$}{d=2}}

Having an exactness guarantee for Algorithm~\aA{} in dimension $d=2$, in our first simulation exercise we compare this procedure with the exact implementation of the algorithm called \code{isodepth} for the computation of central regions for $d=2$ from the \proglang{R} package \pkg{depth} \cite{R_depth}. The latter procedure was originally designed in \cite{Ruts_Rousseeuw1996}, and later revised in \cite{Rousseeuw_Ruts1998, Rousseeuw_etal1999}. It is based on the idea of a circular sequence applied to the angles between pairs of data points. This method is applicable only in dimension $d=2$. %The complexity of \code{isodepth} is $\bigO{n^2\,\log(n)}$.
The efficient \proglang{Fortran} implementation of \code{isodepth} available in the package \pkg{depth} is matched to our \proglang{C++} implementation of Algorithm~\aA{} in the package \pkg{TukeyRegion}. For different values of sample sizes $n \in \left\{ 100, 250, 500, 1000, 2500 \right\}$ we generate random samples from the standard bivariate normal distribution and compute their central regions. The regions considered are first at single levels $k \in \left\{ \floor{n/10}, \floor{n/5}, \floor{n/3} \right\}$ and then also at all levels $k$ at once. In the latter case, we simulate also the task of finding the Tukey median, the smallest non-empty central region. We ran $100$ independent replications of this setup. The comparison was conducted on a machine having processor Intel(R) Core(TM) i7-4980HQ (2.8 GHz) with 16 GB of physical memory and macOS Monterey (Version 12.4) operating system. The resulting execution times (in seconds) are presented in Table~\ref{tab:isodepth}. Since Algorithm~\aA{} is a particular case that implements a general framework for any dimension $d$, it is slightly outperformed by \code{isodepth}, designed only for $d=2$, for smaller sample sizes $n$. On the other hand, Algorithm~\aA{} gains the upper hand over \code{isodepth} already for $n\ge 250$ and with growing sample size becomes even more advantageous.

\begin{table}[ht]
\centering
\begin{tabular}{cc|cc}
\hline & & \multicolumn{1}{c}{\code{TukeyRegion}~\aA{}} & \multicolumn{1}{c}{\code{isodepth}} \\ 
   \hline
  $n = 100$  & $\floor{n/10}$ & \phantom{0}\phantom{0}0.00219  \scriptsize{(0.000510)}    & \phantom{0}\phantom{0}0.00154  \scriptsize{\phantom{0}(0.000302)}    \\ 
             & $\floor{n/5}$  & \phantom{0}\phantom{0}0.00277  \scriptsize{(0.000566)}    & \phantom{0}\phantom{0}0.00205  \scriptsize{\phantom{0}(0.000261)}    \\ 
             & $\floor{n/3}$  & \phantom{0}\phantom{0}0.00294  \scriptsize{(0.000490)}    & \phantom{0}\phantom{0}0.00249  \scriptsize{\phantom{0}(0.000437)}    \\ 
             & All            & \phantom{0}\phantom{0}0.0831\phantom{0}  \scriptsize{(0.00534)\phantom{0}}    & \phantom{0}\phantom{0}0.0773\phantom{0}  \scriptsize{\phantom{0}(0.00871)\phantom{0}}    \\ 
  $n = 250$  & $\floor{n/10}$ & \phantom{0}\phantom{0}0.00491  \scriptsize{(0.000935)}    & \phantom{0}\phantom{0}0.00945  \scriptsize{\phantom{0}(0.000607)}    \\ 
             & $\floor{n/5}$  & \phantom{0}\phantom{0}0.00647  \scriptsize{(0.000451)}    & \phantom{0}\phantom{0}0.0129\phantom{0}  \scriptsize{\phantom{0}(0.00118)\phantom{0}}    \\ 
             & $\floor{n/3}$  & \phantom{0}\phantom{0}0.00795  \scriptsize{(0.000394)}    & \phantom{0}\phantom{0}0.0172\phantom{0}  \scriptsize{\phantom{0}(0.00210)\phantom{0}}    \\ 
             & All            & \phantom{0}\phantom{0}0.651\phantom{0}\phantom{0}  \scriptsize{(0.0128)\phantom{0}\phantom{0}}    & \phantom{0}\phantom{0}1.623\phantom{0}\phantom{0}  \scriptsize{\phantom{0}(0.244)\phantom{0}\phantom{0}\phantom{0}}    \\ 
  $n = 500$  & $\floor{n/10}$ & \phantom{0}\phantom{0}0.0133\phantom{0}  \scriptsize{(0.00115)\phantom{0}}    & \phantom{0}\phantom{0}0.0445\phantom{0}  \scriptsize{\phantom{0}(0.00276)\phantom{0}}    \\ 
             & $\floor{n/5}$  & \phantom{0}\phantom{0}0.0191\phantom{0}  \scriptsize{(0.000900)}    & \phantom{0}\phantom{0}0.0610\phantom{0}  \scriptsize{\phantom{0}(0.00485)\phantom{0}}    \\ 
             & $\floor{n/3}$  & \phantom{0}\phantom{0}0.0240\phantom{0}  \scriptsize{(0.00124)\phantom{0}}    & \phantom{0}\phantom{0}0.0817\phantom{0}  \scriptsize{\phantom{0}(0.00808)\phantom{0}}    \\ 
             & All            & \phantom{0}\phantom{0}4.44\phantom{0}\phantom{0}\phantom{0}  \scriptsize{(0.0418)\phantom{0}\phantom{0}}    & \phantom{0}20.1\phantom{0}\phantom{0}\phantom{0}\phantom{0}  \scriptsize{\phantom{0}(4.53)\phantom{0}\phantom{0}\phantom{0}\phantom{0}}    \\ 
  $n = 1000$ & $\floor{n/10}$ & \phantom{0}\phantom{0}0.0437\phantom{0}  \scriptsize{(0.00234)\phantom{0}}    & \phantom{0}\phantom{0}0.237\phantom{0}\phantom{0}  \scriptsize{\phantom{0}(0.0116)\phantom{0}\phantom{0}}    \\ 
             & $\floor{n/5}$  & \phantom{0}\phantom{0}0.0676\phantom{0}  \scriptsize{(0.00250)\phantom{0}}    & \phantom{0}\phantom{0}0.319\phantom{0}\phantom{0}  \scriptsize{\phantom{0}(0.0148)\phantom{0}\phantom{0}}    \\ 
             & $\floor{n/3}$  & \phantom{0}\phantom{0}0.0873\phantom{0}  \scriptsize{(0.00456)\phantom{0}}    & \phantom{0}\phantom{0}0.416\phantom{0}\phantom{0}  \scriptsize{\phantom{0}(0.0335)\phantom{0}\phantom{0}}    \\ 
             & All            & \phantom{0}33.5\phantom{0}\phantom{0}\phantom{0}\phantom{0}  \scriptsize{(0.236)\phantom{0}\phantom{0}\phantom{0}}    & 257\phantom{0}\phantom{0}\phantom{0}\phantom{0}\phantom{0}\phantom{0}  \scriptsize{(66.6)\phantom{0}\phantom{0}\phantom{0}\phantom{0}\phantom{0}}   \\ 
  $n = 2500$ & $\floor{n/10}$ & \phantom{0}\phantom{0}0.271\phantom{0}\phantom{0}  \scriptsize{(0.0225)\phantom{0}\phantom{0}}    & \phantom{0}\phantom{0}2.83\phantom{0}\phantom{0}\phantom{0}  \scriptsize{\phantom{0}(0.0552)\phantom{0}\phantom{0}}    \\ 
             & $\floor{n/5}$  & \phantom{0}\phantom{0}0.429\phantom{0}\phantom{0}  \scriptsize{(0.0154)\phantom{0}\phantom{0}}    & \phantom{0}\phantom{0}3.44\phantom{0}\phantom{0}\phantom{0}  \scriptsize{\phantom{0}(0.135)\phantom{0}\phantom{0}\phantom{0}}    \\ 
             & $\floor{n/3}$  & \phantom{0}\phantom{0}0.614\phantom{0}\phantom{0}  \scriptsize{(0.0195)\phantom{0}\phantom{0}}    & \phantom{0}\phantom{0}4.18\phantom{0}\phantom{0}\phantom{0}  \scriptsize{\phantom{0}(0.253)\phantom{0}\phantom{0}\phantom{0}}    \\ 
             & All            & 532\phantom{0}\phantom{0}\phantom{0}\phantom{0}\phantom{0}\phantom{0}  \scriptsize{(2.55)\phantom{0}\phantom{0}\phantom{0}\phantom{0}}   & $>1$ hour \\ 
   \hline
\end{tabular}
\caption{Means and standard deviations (in brackets) of execution times for algorithms~\aA{} and~\code{isodepth} when calculating depth contour(s) for a bivariate normal sample (in seconds, over $100$ random samples).}
\label{tab:isodepth}
\end{table}

%%\problem{Table will be improved, and better explained in the caption.}

%%\problem{[PAVLO: Can you please run the full simulation study of Algorithm~\aA{} against isodepth? The results above are only for I think 10 independent runs on my computer.]}

%%\problem{[PAVLO: Is there a reason why for small $n$ isodepth is faster than (A)?]}

%%\pavlo{Yes, because my implementation is larger and more ``bulky'', after all I did not implement it for $d=2$ only, just optimised memory allocation where could...}

\subsection{Algorithm~\aB{} and its competitors for \texorpdfstring{$d>2$}{d>2}}    \label{sec:simulation2}

We have implemented the exact Algorithm~\aB{} in the new version 0.1.5.5 of the \proglang{R} package \pkg{TukeyRegion} \cite{R_TukeyRegion}. We compare three algorithms for the computation of the central regions, each based on~\textbf{\textsf{RidgeSearch}}: \begin{enumerate*}[label=(\roman*)] \item the non-exact fast Algorithm~\aA{} with the initialisation~\ref{A1}--\ref{A3}; \item our new exact Algorithm~\aB{}; and \item the exact combinatorial Algorithm~\aCmb{} based on plugging all $\binom{n}{d-1}$ ridges of $X$ into the initial queue $\mathcal Q$. \end{enumerate*} The complexity of Algorithm~\aB{} amounts to $\bigO{\omega\, n^d \log(n)}$. Here $\omega$ depends on the task at hand. When a constant number of outer regions is to be calculated, its complexity can be as low as $\bigO{1}$. When, on the other hand, one is interested in computing a portion (say $\gamma\in[0,1/2]$) of outer central regions, the complexity of $\omega$ is $\bigO{n}$. Consequently, computation of all regions has time complexity $\bigO{n^{d+1} \log(n)}$; for $d=2$ this yields the same complexity as \code{isodepth}, namely $\bigO{n^3 \log(n)}$. 

For several combinations of $n$ and $d$ we have drawn $100$ independent samples $X$ from the $d$-variate standard normal distribution of size $n$. Since such samples $X$ are in general position almost surely, the maximum Tukey depth of a full-dimensional central region is bounded from above by $\floor{(n - d + 1)/2}$, see \cite[Theorem~1]{Liu_etal2020}. For each $k \in \left\{ 1, \dots, \floor{(n - d + 1)/2} \right\}$ we computed the first $k$ central regions $\HD_{\ell}(X)$, $\ell = 1, \dots, k$, using the Algorithms~\aA{}--\aCmb{}. We kept a record of \begin{enumerate*}[label=(\roman*)] \item the number of found relevant halfspaces $\half_k$ at level $k$; \item the number of ridges visited by the algorithm; and \item the total execution time for computing the first $k$ central regions of the data. \end{enumerate*} The full study was run for the following combinations of $(n,d)$: $(50,3)$, $(50,4)$, $(50,5)$, $(100,3)$, $(100,4)$, and $(250,3)$. With a smaller number of independent runs, we have tested the algorithms also with larger values of $n$ up to $n=5\,000$ for $d=3$ and $n=1\,000$ for $d=4$; the results are quite analogous to those presented below. 

First, we evaluated the exactness of Algorithm~\aA{}. In the complete numerical study, we generated $600$ normal samples, and calculated a total of $(24+23+23+49+48+124)\times 100=29\,100$ central regions. Out of these results, only in the case of a single central region (for $d=3$ and $n=50$) Algorithm~\aA{} failed to detect all relevant halfspaces --- similarly as in our example from Section~\ref{sec:not exact}, four relevant halfspaces were missed. This result validates the empirical evidence from \cite{Liu_etal2019}. For random samples from well behaved distributions, Algorithm~\aA{} is quite likely to give exact results; nevertheless, in general it is non-exact.

\begin{figure}[htpb]
    \centering
    \includegraphics[width=.85\textwidth]{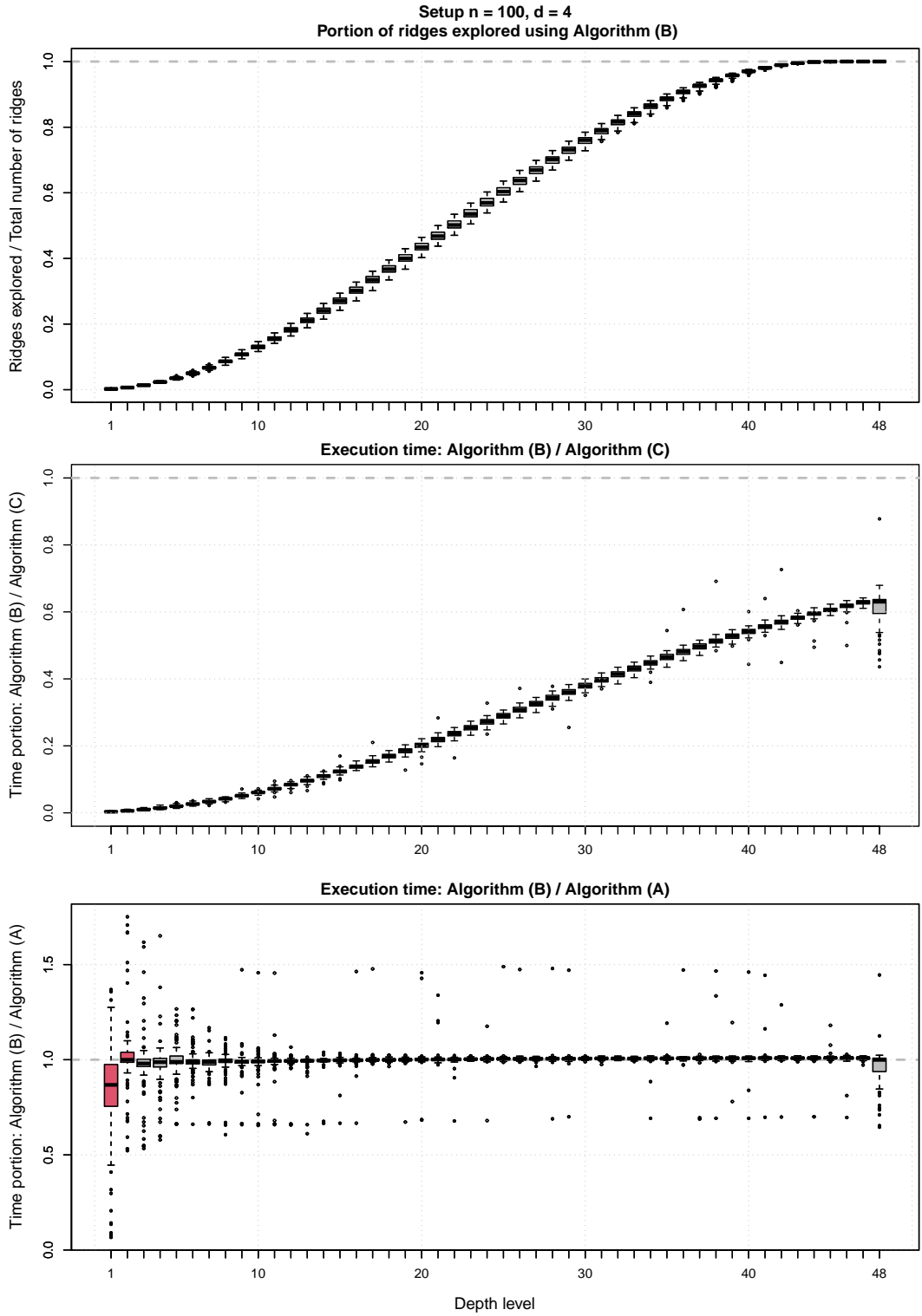} %\\
    \caption{Results of the numerical study, setup $n = 100$, $d = 4$. For a detailed description of the results see Section~\ref{sec:simulation2}.}
    \label{fig:all}
\end{figure}

In this section we discuss the detailed results of the numerical study for $d = 4$ and $n = 100$. The results are summarised in Figure~\ref{fig:all}, which consists of three parts: \begin{enumerate*}[label=(\roman*)] \item In the top panel we see the boxplots of the proportion of ridges visited by Algorithm~\aB{}, compared to the total number of $\binom{n}{d-1}$ ridges of $X$ used by Algorithm~\aCmb{}. We see that this fraction is, especially for lower to moderate depth levels $k$, relatively small. \item In the middle panel there are the boxplots of ratios of execution times of Algorithms~\aB{} and~\aCmb{}. This figure resembles the one from the top panel. Even for the complete set of central regions (that is, all regions up to $k = \floor{(n - d + 1)/2} = 48$), the total computation time for Algorithm~\aB{} typically does not exceed $70~\%$ of the time used by Algorithm~\aA{}. For computation of lower regions the spared computation power is substantial. \item Finally, in the bottom panel we see an analogous display with a fraction of computation time, this time Algorithm~\aB{} compared to the non-exact Algorithm~\aA{}. The first two boxplots ($k=1,2$) may be disregarded as in that case Algorithm~\aA{} is exact. In all other boxplots, we see that in addition to having a guarantee of exactness for Algorithm~\aB{}, our computation in Algorithm~\aB{} does not increase the execution times of Algorithm~\aA{}. \end{enumerate*} All these results are quite favourable. In the task of computing all the first $k$ central regions of $X$, the new Algorithm~\aB{} is more efficient than the combinatorial Algorithm~\aCmb{}, and also not slower than the fast Algorithm~\aA{}. The final results for other combinations of $n$ and $d$ are quite similar, and given without additional commentary in the Supplementary Material. To get a rough idea about the raw computation times of the three algorithms, in the Supplementary Material we include a table with average execution times of the three considered algorithms, matching the more extensive simulation study presented in~\cite[Tables~2 and~3]{Liu_etal2019}. 

In our numerical study we computed all central regions at levels $1, \dots, k$ at the same time. It is, however, important to mention that in contrast to the combinatorial Algorithm~\aCmb{}, Algorithm~\aB{} is inherently recursive when initialising the queue $\mathcal Q$ in~\ref{B1}. Therefore, Algorithm~\aB{} is typically slower than the direct combinatorial Algorithm~\aCmb{} if the computation of a single central region is of interest. We therefore conclude that for exact computation of Tukey depth central regions, Algorithm~\aB{} is the fastest if all regions $\HD_k(X)$ are to be computed for $k=1,\dots,K$, or if a region at a lower level $k$ is to be found. Algorithm~\aCmb{} is to be preferred if a single region $\HD_k(X)$ for a higher value $k$ is searched for.

%
%
% Counterexamples in the dual space
%
%

\section{Dual graph: (Non-)Exactness of Algorithm~\aA{} and negative results}  \label{sec:dual}

We saw that Algorithm~\aA{} does not always find the central region. On the other hand, the exact Algorithm~\aB{} proposed in Section~\ref{sec:exact} is typically slower than Algorithm~\aA{} if a single central region $\HD_k(X)$ is to be evaluated. In the present section we first explore the negative example of Section~\ref{sec:not exact} in view of polarity considerations. We present a different vantage point on the search for relevant halfspaces based on the duality theory. It is shown that in the so-called dual graph of the dataset $X$, the search for an exact algorithm manifests itself as a natural problem in the theory of graphs. This analysis serves us to show that multiple promising simplifications of our Algorithm~\aB{} along the lines of Algorithm~\aA{} fail to recover the exact Tukey depth central region.

%
% Dual graph
%

\subsection{Dual graph: Definition}

We have seen in Section~\ref{sec:polarity} that any polytope $P$ in $\R^d$ whose interior contains the origin can be represented in its dual form as a polytope $P^\circ$. The facets $F_j$ of $P$ correspond to the vertices $\hat{F}_j$ of $P^\circ$, and facets $F_j$ and $F_k$ of $P$ are mutually neighbouring if and only if the vertices $\hat{F}_j$ and $\hat{F}_k$ share an edge on $P^\circ$. Instead of working with polytopes, we now generalise the polarity paradigm directly toward a dataset $X$. To visualise our problem of finding all relevant halfspaces at a given level $k = 1,2,\dots$ of $X$, we introduce the dual graph of $X$, and show how the search strategy employed in~\textbf{\textsf{RidgeSearch}} translates into a problem of graph connectivity in the dual space.

Any $d$ distinct points $a_1, \dots, a_{d}$ from $X$ uniquely determine a hyperplane $\aff{a_1, \dots, a_d}$. We suppose that $X$ is such that none of these $\binom{n}{d}$ hyperplanes passes through the origin; in the other case we shift the dataset $X$ slightly. The \emph{polar} to $\aff{a_1, \dots, a_d}$ is defined as 
    \begin{equation}    \label{polar hyperplane}
     \aff{a_1, \dots, a_d}^\circ = \left\{ x \in \R^d \colon \left\langle x, y \right\rangle = 1 \mbox{ for all } y \in \aff{a_1, \dots, a_d} \right\}.
    \end{equation}
%\problem{Should we further assume that the points $X \cup \{0\}$ are in general position? This implies that in the dual space each $d$-tuple of hyperplanes (corresponding to $d$ points) intersect in a point. Otherwise, for example in $\R^2$ the lines that correspond to points $(-1, 0)$ and $(1, 0)$ do not intersect and the point $[(-1, 0), (1, 0)]^\circ$ would not be defined. -VF}
%
This definition is analogous to that of a conjugate face from \eqref{conjugate face}. Indeed, for a face $F = \conv{a_1, \dots, a_d}$ we have $\hat{F} = \aff{a_1, \dots, a_d}^\circ$. In particular, each polar to an observational hyperplane is a single point in the dual space. The polar from \eqref{polar hyperplane} is easy to express analytically. Writing $u \in \Sph$ for a unit normal vector of $\aff{a_1, \dots, a_d}$, we have 
    \begin{equation}    \label{dual vertex}
     \aff{a_1, \dots, a_d}^\circ = \left\{ \frac{u}{\left\langle a_1, u \right\rangle} \right\}.
    \end{equation}  
Note that by our assumption that $\aff{a_1, \dots, a_d}$ does not contain the origin, the single point set above is well defined as $\left\langle a_1, u \right\rangle$ is the distance of the hyperplane $\aff{a_1, \dots, a_d}$ from the origin. Analogously as in Section~\ref{sec:polarity}, for two different observational hyperplanes $\aff{a_1, \dots, a_d}$ and $\aff{b_1, \dots, b_d}$ we join the pair of vertices $\aff{a_1, \dots, a_d}^\circ$ and $\aff{b_1, \dots, b_d}^\circ$ in the dual space by an edge if and only if the two hyperplanes are \emph{(mutually) neighbouring}, meaning that their defining sets of data points share a common subset of exactly $d-1$ elements (a ridge)
    \[  \# \left(\left\{a_1, \dots, a_d\right\} \cap \left\{b_1, \dots, b_d\right\} \right) = d-1.  \]
This definition is an extension of the notion of neighbouring facets of a polytope from Section~\ref{sec:k2} with $k=2$. Finally, each vertex \eqref{dual vertex} in the dual space corresponding to a hyperplane $\aff{a_1, \dots, a_d}$ is assigned a weight $k$ being the smaller number of points from $X$ that are cut off by $\aff{a_1, \dots, a_d}$ plus one, that is
    \[  
    \begin{aligned}
    k = 1 + \min \left\{ \# \left\{X \cap \left\{x \in \R^d \colon \frac{\left\langle x, u \right\rangle}{\left\langle a_1, u \right\rangle} > 1 \right\} \right\}, \# \left\{X \cap \left\{x \in \R^d \colon \frac{\left\langle x, u \right\rangle}{\left\langle a_1, u \right\rangle} < 1 \right\} \right\} \right\}.
    \end{aligned}
    \]
This weight corresponds to the level $k$ at which $\aff{a_1, \dots, a_d}$ forms a boundary of some $H \in \half(k)$. Altogether, in the dual space we obtain a graph \begin{enumerate*}[label=(\roman*)] \item whose $\binom{n}{d}$ vertices correspond to observational hyperplanes, and \item each such vertex is given a weight corresponding to the level $k$ at which it contributes to $\HD_k(X)$. \item Two vertices are joined by an edge if and only if the corresponding hyperplanes share a ridge. \end{enumerate*} We call this the \emph{dual graph} of $X$.

For the following analysis of $G$ we involve tools from the theory of graphs. An \emph{induced subgraph} $G'$ of $G$ is a graph formed by a subset $V'$ of the set of vertices $V$ of $G$, and all the edges of $G$ that join pairs of points from $V'$. An induced subgraph $G'$ of $G$ is a called a \emph{clique} if each pair of vertices of $G'$ is connected by an edge. It is a \emph{maximal clique} if no other vertex of $G$ can be appended to $G'$ so that the resulting induced subgraph is a clique. 

The dual graph $G$ describes all the relevant combinatorial structure of our arrangement of points $X$ --- the observational hyperplanes, the fact whether two such hyperplanes are neighbouring, and also the fact whether they share a ridge. Indeed, to identify common ridges, note that two observational hyperplanes share a ridge $I$ if and only if their dual vertices are connected in $G$ by an edge. Any given ridge $I$ is shared by $n-(d-1)$ hyperplanes, meaning that the corresponding $n-(d-1)$ vertices of $G$ form a clique in $G$. Since no other observational hyperplane contains $I$, that clique is maximal. Consequently, each ridge $I$ is equivalent with a maximal clique of $G$, and each such maximal clique has $n-d+1$ vertices.

%
% Example from before
%

\subsection{Algorithm~\texorpdfstring{\aC{}}{A2}: The counterexample revisited} %\label{sec:not exact graph}
We illustrate the relevance of the dual graph by returning to our example from Section~\ref{sec:not exact}. The dual graph of the set of $n=12$ points $X$ in $\R^3$ has $\binom{n}{d} = \binom{12}{3} = 220$ vertices, each connected with exactly $d\,(n-d) = 27$ other vertices. The weights of the vertices range from $k = 1$ for those corresponding to hyperplanes forming the boundary of $\conv{X}$, to $k=\lfloor(n-d)/2\rfloor + 1 = 5$. The latter bound follows from e.g. the ham sandwich theorem \cite{Elton_Hill2011, Matousek2003}. 

\begin{figure}[htpb]
    \centering
    $\vcenter{\hbox{\fbox{\includegraphics[height=.2\textwidth]{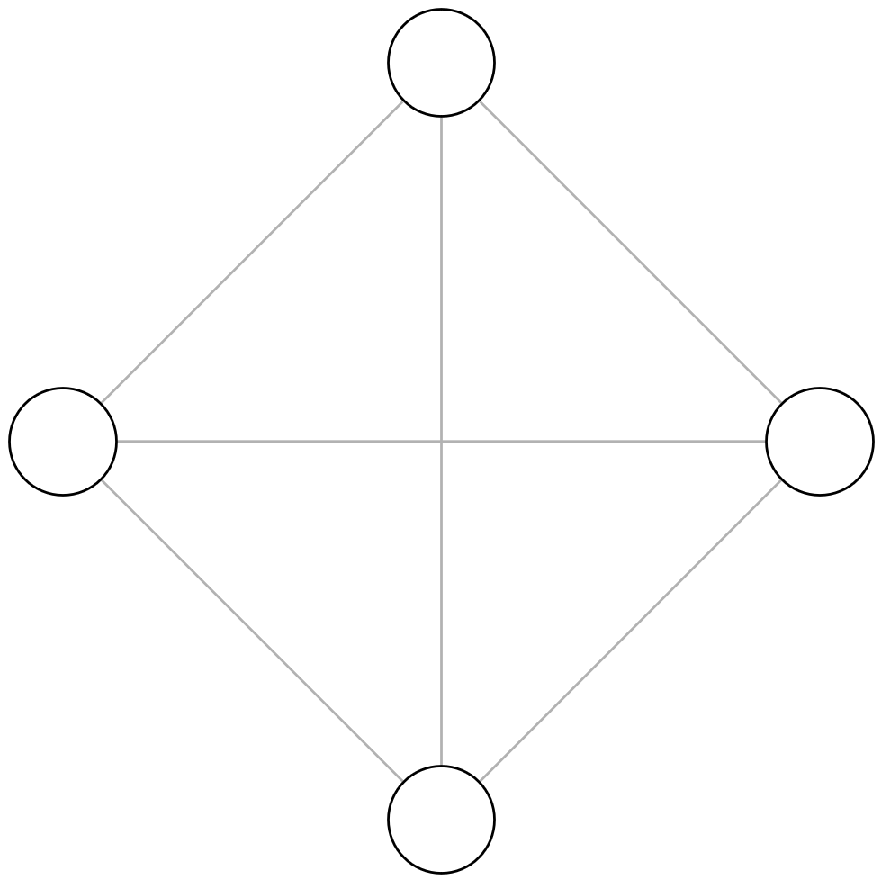}}}}$ \quad    
    $\vcenter{\hbox{\fbox{\includegraphics[height=.3\textwidth]{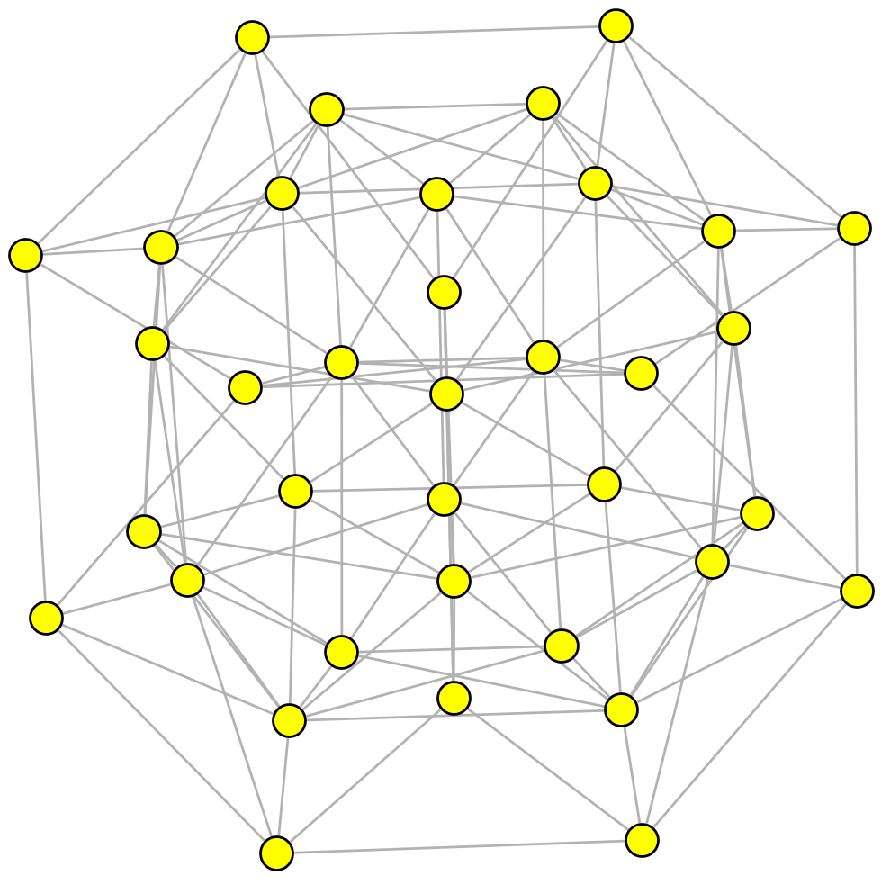}}}}$ \quad
    $\vcenter{\hbox{\fbox{\includegraphics[height=.35\textwidth]{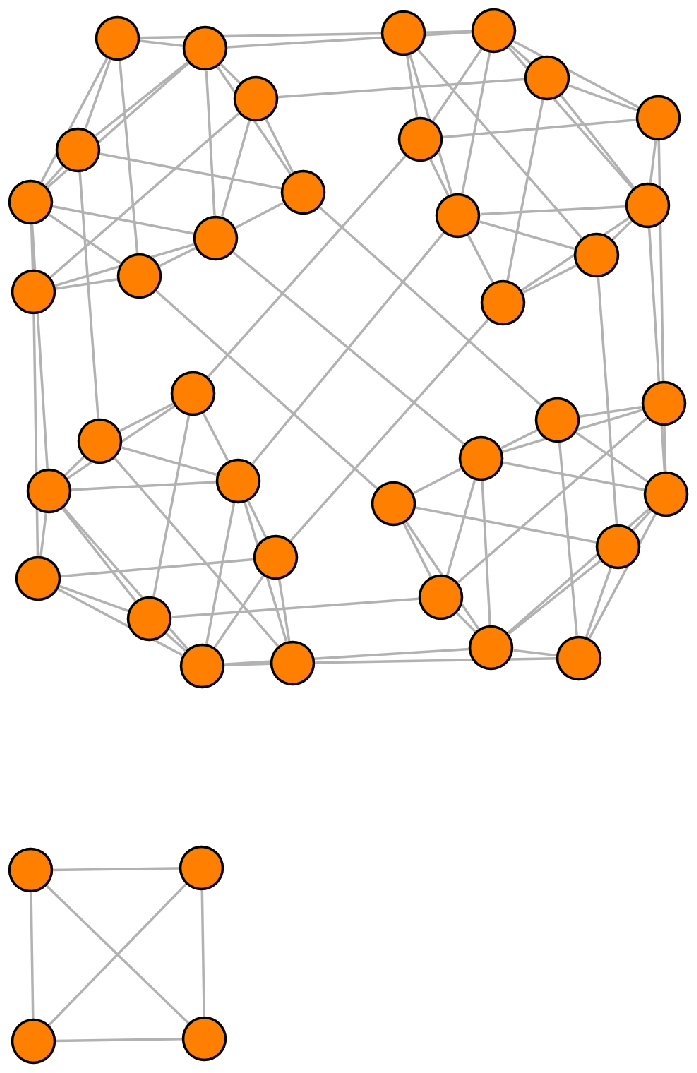}}}}$ 
    \caption{Induced monochrome subgraphs $G_k$ (that is, subgraphs of vertices of different weights) of the dual graph of $X$ from Section~\ref{sec:not exact}: $k = 1$ in the left hand panel, $k = 2$ in the middle panel, and $k = 3$ in the right hand panel.}
    \label{fig:dual subgraph}
\end{figure}

Its complexity makes direct visualisation of $G$ cumbersome; the complete dual graph of $X$ can be found in the online Supplementary Material. It turns out that it is more insightful to restrict to the induced monochrome (that is single-weight) subgraphs of $G$. Those are displayed in Figure~\ref{fig:dual subgraph}, where we can see the subgraphs $G_1$, $G_2$, and $G_3$ of hyperplanes that cut off $0$, $1$, or $2$ points from $X$, respectively. The vertices of the graph $G_k$ represent all the elements of $\half(k)$. We identify two connected components of $G_3$ corresponding to hyperplanes from $\half(3)$. The smaller connected component relates to the four faces of the blue tetrahedron from Figure~\ref{fig:12points}. We have seen in Section~\ref{sec:not exact} that these relevant hyperplanes are never reached using Algorithm~\aA{}. Now we argue that this difficulty is fundamental, and even a substantially expanded procedure~\textbf{\textsf{RidgeSearch}} in the spirit of Algorithm~\aA{} still fails to find all halfspaces from $\half(3)$. Consider the following variation of Algorithm~\aA{} where, to search for all elements of $\half(k)$ (that is, all vertices of $G_{k}$), we make the following amendments:
    \begin{enumerate}[label=\textbf{(A$_{\arabic*}^2$)},ref=\upshape{\textbf{(A$_{\arabic*}^2$)}}]
        \item \label{A12} In \ref{A1}, only a \emph{single ridge} $I$ of points on the convex hull of $X$ is considered. Instead, we take \emph{all ridges} of all facets of $\conv{X}$ in the initial step.
        \item In Algorithm~\aA{}, only \emph{two hyperplanes} relevant at level $k$ that share a given ridge $I$ are pushed into the queue $\mathcal Q$ in \textbf{Step~2(d)} in \cite{Liu_etal2019}. We push \emph{all hyperplanes} relevant at level $k$ that share $I$ into the search queue $\mathcal Q$, as we did in~\ref{A2}.
        \item \label{A32} Finally, in \ref{A3}, after $H \in \half(k)$ that contains a \emph{given ridge $I$ of the facet of $\conv{X}$} is found, each \emph{point of $X$ that was cut off} by $H$ from $X$ is combined with points of the ridge $I$ to form new ridges pushed to $\mathcal Q$. In our setting, we expand the last set of ridges considerably. We add to $\mathcal Q$ all the ridges of all hyperplanes that are \emph{relevant at all levels $\ell = 2, 3, \dots, k$} and share \emph{at least a single ridge with any facet} of $\conv{X}$.
    \end{enumerate}
The steps~\ref{A12}--\ref{A32} can be summarized more succinctly in dual terms. To search for the central region at level $k$, all ridges of
    \begin{itemize}
        \item[\ding{228}] all vertices of $G_1$, and
        \item[\ding{228}] all vertices of $G_2, \dots, G_{k}$ that share an edge with any vertex from $G_1$,
    \end{itemize} 
are added to $\mathcal Q$ in the initial step~\ref{S1} of~\textbf{\textsf{RidgeSearch}}. Our expanded program then proceeds by~\ref{S2} and~\ref{S3}. In plain words, for the graph $G$ it means that instead of through ridges in $\mathcal Q$ we launch a search through relevant hyperplanes represented as the vertices of $G$ (and all ridges of those hyperplanes). Instead of a queue of ridges $\mathcal Q$ we take a queue $\mathcal V$ of vertices of $G$. After initialising $\mathcal V$, in~\ref{S2} all connected components of $G_k$ that share an edge with a vertex from the initial set $\mathcal V$ are found, and their vertices are added to $\mathcal V$. Our algorithm is exact if the resulting $\mathcal V$ contains all vertices of $G_k$.

The expanded procedure based on \ref{A12}--\ref{A32} and~\textbf{\textsf{RidgeSearch}} is called Algorithm~\aC{} for brevity. The search for all relevant hyperplanes is easy to visualise in the dual graph $G$ of $X$. The results of applying both Algorithms~\aA{} and~\aC{} to the dataset $X$ from Section~\ref{sec:not exact} for $k=3$ are found in Figure~\ref{fig:AlgA}. None of these algorithms recovers all relevant hyperplanes of $X$; they both miss the same blue tetrahedron from Figure~\ref{fig:12points}. 

\begin{figure}[htpb]
    \centering
    \fbox{\includegraphics[width=.35\textwidth]{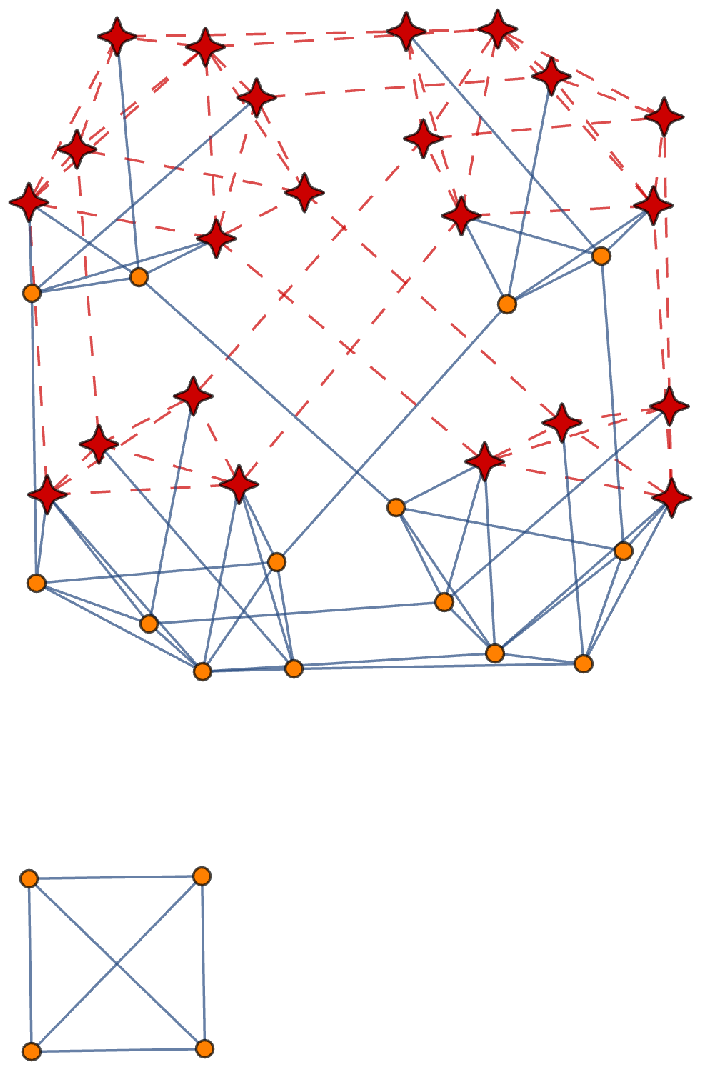}} \quad
    \fbox{\includegraphics[width=.35\textwidth]{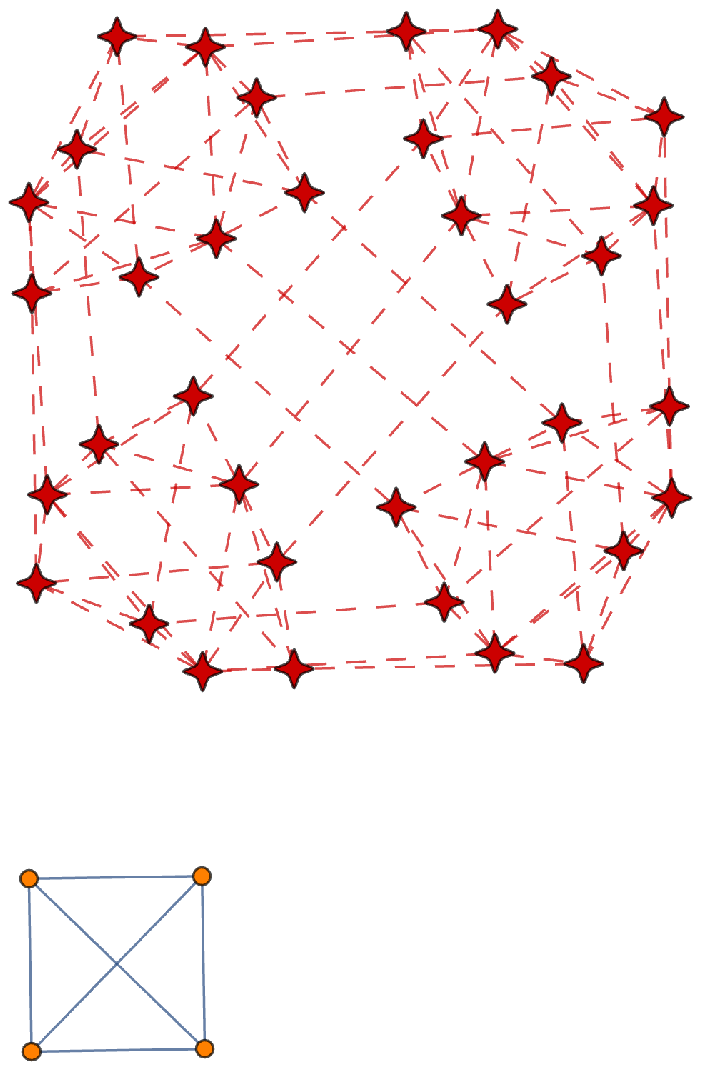}} 
    \caption{The induced subgraph $G_3$ in the example from Section~\ref{sec:not exact}. In both graphs, the vertices highlighted as diamonds are those vertices of $G_3$ that are found in the initial step of~\textbf{\textsf{RidgeSearch}} using Algorithm~\aA{} (left hand panel) and its expanded version Algorithm~\aC{} (right hand panel). Neither of these algorithms detects the tetrahedron of blue points from Figure~\ref{fig:12points}, which corresponds to the smaller connected component of $G_3$ in these figures.}
    \label{fig:AlgA}
\end{figure}

%\subsection{All relevant halfspaces that contain a given point.} \problem{In the initial step, we consider as a starting set of hyperplanes all relevant hyperplanes that contain a given point (possibly in its interior).}

%
%
% Appendix
%
%

\appendix

\section{Proofs of theoretical results}

\subsection{Proof of Theorem~\ref{thm:orbitspassthrough}}   \label{sec:orbits}

For a relevant halfplane $H \in \half(k)$ with $L = \aff{x_i, x_j} = \partial H$ we denote by $\nu_L = \nu_{\aff{x_i,x_j}} \in\Sph[1]$ the outer unit normal vector of $H$. In other words, the vector $\nu_L = \nu_{\aff{x_i,x_i}}$ is orthogonal to $L = \aff{x_i, x_j}$ and heads towards the open halfplane that contains exactly $k-1$ points from $X$. For unit vectors $u_1, u_2 \in \Sph[1]$ define the closed spherical interval $[u_1, u_2]$ to be the shorter arc of unit vectors $\Sph[1]$ delimited by $u_1$ and $u_2$, including its endpoints. In the case $u_1 = - u_2$ any of the two half-circles between $u_1$ and $u_2$ can be chosen as $[u_1,u_2]$. The open spherical interval $(u_1,u_2)$ is defined analogously.

For the proof of Theorem~\ref{thm:orbitspassthrough} we need the following lemma.

\begin{lemma}
\label{lm:closerpoint}
Let $L = \partial H_L$ for $H_L \in \half(k)$ and let $\aff{x_i, x_j}$ be another relevant line at level $k$ such that $x_i, x_j \in H_L \setminus L$. That is, both points $x_i, x_j$ are in the interior of the halfplane $H_L$. Then there exists a relevant line $L^\prime$ at level $k$ that passes through either $x_i$ or $x_j$ such that $\nu_{L^\prime} \in (\nu_{\aff{x_i, x_j}}, \nu_{L})$.
\end{lemma}

\begin{proof}
Assume without loss of generality that $x_i$ is closer or equally far from $L$ compared to $x_j$. Let $\nu \colon [0, 1] \to [\nu_{\aff{x_i, x_j}}, \nu_{L}]$ be any continuous bijection such that $\nu(0) = \nu_{\aff{x_i, x_j}}$ and $\nu(1) = \nu_{L}$. For $s \in [0,1]$, let $L_s$ be the line with normal vector $\nu(s)$ passing through $x_i$ and let $h(s)$ be the number of points in the open halfplane $G_s$ defined by $L_s$ for which $\nu(s)$ is an inner normal vector. This situation is visualised in Figure~\ref{fig:Lemma2}.

\begin{figure}
    \centering
    \includegraphics[width=.40\textwidth]{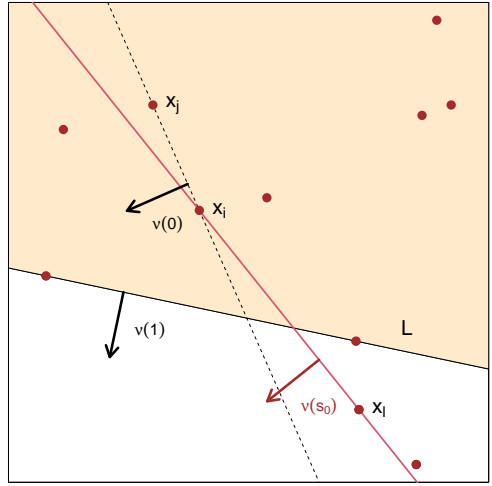}
    \caption{The situation in the proof of Lemma~\ref{lm:closerpoint} for $k=3$. By rotating the dashed line $\aff{x_i, x_j}$ around the point $x_i$ we necessarily reach another relevant line $L^\prime = \aff{x_i, x_l}$ (brown line) whose unit normal $\nu(s_0)$ lies in $\Sph[1]$ strictly between $\nu(0)$ and $\nu(1)$, and at the same time $L'$ cuts off exactly $k-1$ points from $X$.}
    \label{fig:Lemma2}
\end{figure}
	
Then, $h\colon [0,1] \to \{0,1,\dots\}$ is a lower semi-continuous step function with steps of size~$1$, because we suppose that $X$ lies in general position. A step of $h$ occurs exactly at each $s$ such that $L_s$ passes through another point from $X$. For $\varepsilon_1 > 0$ small enough, $h(s) = k - 1$ for all $s \in [0,\varepsilon_1)$. Furthermore, since $L_1$ is parallel with $L$ and hence $G_1$ contains both $L$ and the $k-1$ points from $X$ cut off by $L$, we have $h(1) \ge 2+k-1 = k+1$. Take $s_0 = \sup \{s \in [0,1] \colon h(s) = k - 1 \}$. Because $h(s) = k-1$ for $s$ small enough, we have $s_0 > 0$. We have $h(1) \ge k+1$, $h(s) \neq k-1$ for all $s \in (s_0, 1]$, and $h$ has only steps of size $1$. Necessarily, there must exist a point $x_l \in X$ such that $x_l \in L_{s_0}$, $h(s_0) = k - 1$, and $\nu(s_0) \in (\nu_{\aff{x_i, x_j}}, \nu_{L})$. So, $L^\prime = \aff{x_i, x_l}$ is a relevant line satisfying the claim.
%
% there exists some $\varepsilon_2 > 0$ small enough such that for $s \in (s_0, s_0 + \varepsilon_2]$ we have $h(s) = k$. That is, there exists a point $x_l \in X$ that is an element of the open halfplane $G_{s_0+\varepsilon}$ corresponding to $L_{s_0 + \varepsilon}$, but not of the one corresponding to $L_{s_0 - \varepsilon}$ for all $\varepsilon \in (0,\varepsilon_2)$. Therefore, $x_l \in L_{s_0}$, $h(s_0) = k - 1$ and $\nu(s_0) \in (\nu_{\aff{x_i, x_j}}, \nu_{L})$. So, $L^\prime = \aff{x_i, x_l}$ is a relevant line satisfying the claim.
\end{proof}

We are now ready to prove Theorem~\ref{thm:orbitspassthrough}. 

\smallskip
\noindent
\textbf{Main proof of Theorem~\ref{thm:orbitspassthrough}.} Let $L = \aff{x_i, x_j} = \partial H$ and $V = U \cup \{x_i, x_j \}$. Choose any $G \in O$ with $\aff{x_{i_0}, x_{i_1}} = \partial G$, where $x_{i_1}$ is no further from $L$ than $x_{i_0}$. If $x_{i_0}, x_{i_1} \notin V$, by Lemma~\ref{lm:closerpoint}, there exists a point $x_{i_2}$ such that $\aff{x_{i_1}, x_{i_2}}$ is a relevant line at level $k$ and $\nu_{\aff{x_{i_1}, x_{i_2}}} \in (\nu_{\aff{x_{i_0}, x_{i_1}}}, \nu_{L})$. If $x_{i_2} \notin V$, continue to construct $x_{i_3}$ and so on in a similar manner. Since the angle between $\nu_{\aff{x_{i_h}, x_{i_{h+1}}}}$ and $\nu_L$ decreases strictly with $h = 0,1,\dots$, the points in the sequence $x_{i_0}, x_{i_1}, x_{i_2}, \dots$ cannot repeat, and because there is only a finite number of points in $X$, the construction cannot continue to infinity. Consequently, there is some $h$ such that $x_{i_h} \in V$. The point $x_{i_h}$ was obtained in such a way that the relevant halfplane $H_{\aff{x_{i_{h-1}}, x_{i_h}}} \in \half(k)$ with boundary $\aff{x_{i_{h-1}}, x_{i_h}}$ is reachable from the relevant halfplane $G$ with boundary $\aff{x_{i_0}, x_{i_1}}$.
		
If $x_{i_{h}} \in \{x_i, x_j \}$, then $H$ is reachable from $G$, and thus $G$ is in the orbit of $H$. Otherwise, $x_{i_{h}} \in U$, meaning $\aff{x_{i_{h-1}}, x_{i_h}}$ is a relevant line passing through $x_{i_{h}} \in U$ whose relevant halfplane $H_{\aff{x_{i_{h-1}}, x_{i_h}}}$ lies in the same orbit as $G$, as we wanted to show.

\subsection{Proof of Lemma~\ref{lm:connected}}  \label{section:connected}

Since $X$ is in general position, facets of $C$ are $(d-1)$-simplices \cite[page~8]{Ziegler1995} and two facets are neighbouring if and only if they share $d-1$ vertices of $C$. Without loss of generality suppose that the origin is an interior point of $C$, and consider the polar $C^\circ$ of $C$ from in Section~\ref{sec:polarity}. We know that $C^\circ$ is also a convex polytope \cite[Lemma~2.4.5]{Schneider2014}, its vertices correspond to facets of $C$ and two vertices are connected by an edge if and only if the corresponding facets of $C$ are neighbours \cite[Section~2.4]{Schneider2014}. All pairs of vertices of a polytope are connected by a sequence of edges \cite[Section~3.5]{Ziegler1995}.

\subsection{Proof of Theorem~\ref{theorem:k2}}  \label{section:k2}

The proof of the main theorem is obtain by combining several auxiliary lemmas. Our first observation is that any two relevant hyperplanes that cut off the same single point of $X$ lie in the same orbit of $\half(2)$.

%lemma: cut off the same point => reachable
\begin{lemma}
	\label{lm:samepoint}
If two relevant halfspaces in $\half(2)$ cut off the same point of $X$, they are mutually reachable in $\half(2)$.
\end{lemma}

\begin{proof}
Let $x_1$ be the point cut off by both halfspaces and consider the convex hull $C_1$ of $X \setminus \{x_1 \}$. Without loss of generality we may suppose that the origin lies in the interior of $C_1$. Each of the facets $F_1, \dots, F_m$ of $C_1$ contains exactly $d$ points from $X \setminus \{x_1 \}$. Using the dual construction \eqref{conjugate face} we find to each $F_j$ a vertex $\hat{F}_j$ of $C_1^\circ$. By duality considerations, for any pair $\hat{F}_j$ and $\hat{F}_k$ we have that the hyperplanes $\aff{F_j}$ and $\aff{F_k}$ are neighbouring in $X \setminus \{x_1 \}$ (that is, sharing $d-1$ points of $X \setminus \{x_1\}$) if and only if $\hat{F}_j$ and $\hat{F}_k$ share an edge on the boundary of the polar polytope $C_1^\circ$.

The boundary $\partial H$ of each relevant halfspace $H \in \half(2)$ that cuts off $x_1$ is determined by $d$ points from $X \setminus \{x_1 \}$. Thus, it contains a facet of $C_1$, and this facet corresponds to a vertex of $C_1^\circ$ by \eqref{conjugate face}. Denote by $U = \left\{ \hat{F}_1, \dots, \hat{F}_m \right\}$ the set of all vertices of $C_1^\circ$, and consider its subset $U'$ of those vertices that correspond to faces of $C_1$ determined by relevant hyperplanes that cut off $x_1$. Vertices $u \in U'$ are specific among those from $U$ by their property that we have $\left\langle x_1, u \right\rangle > 1$. This is because for $u \in U'$ corresponding to a relevant halfspace $H$ that cuts off $x_1$ we can write by \eqref{P by duality} that 
	%\[	
	$\R^d \setminus H = \left\{ x \in \R^d \colon \left\langle x, u \right\rangle > 1 \right\}$	% \]
and $x_1 \notin H$. Consider the set $H_1 = \left\{ x \in \R^d \colon \left\langle x_1, x \right\rangle > 1 \right\}$. In the dual space, this is an open halfspace with inner normal $x_1$ at the distance $x_1/\left\Vert x_1 \right\Vert^2$ from the origin. This halfspace is well defined, as we assumed that in the primal space, the origin is contained in $C_1$ and $x_1 \notin C_1$, i.e. $\left\Vert x_1 \right\Vert > 0$. Using this interpretation, we see that we can write $U' = U \cap H_1$. In particular, $U'$ can be obtained by intersecting $U$ with an open halfspace in the dual space. 

We have reduced our problem to the question of whether any two vertices from $U' = U \cap H_1$ can be joined by a sequence of edges of the polytope $C_1^\circ$ joining two points from $U'$. That is known to be true by a result from the theory of graphs called the Balinski theorem \cite[proof of the main Theorem]{Balinski1961}, stated explicitly also in \cite[Section~3]{Sallee1967}.
\end{proof}

Having established that all relevant halfspaces in $\half(2)$ cutting off the same vertex of $C$ belong to a single orbit, we now explore the structure of orbits in $\half(2)$ with respect to the vertices of $C$ that are cut off. To do this, we need to introduce another notion of reachability, called \emph{vertex-reachability}, suited for the vertices of a convex polytope instead of for halfspaces. This should not be confused with the relation of vertices of $C$ being connected via a sequence of edges of $C$. Our definition is given analogously to the reachability of halfspaces introduced in Section~\ref{sec:notations}. We say that two vertices $x$, $x'$ of $C$ are \emph{(mutually) vertex-reachable} if \begin{enumerate*}[label=(\roman*)] \item they are \emph{vertex-neighbouring}, meaning there exists a set $V$ of $d-1$ vertices of $C$ different from $x$ and $x'$ such that both $\conv{V \cup \{x \}}$ and $\conv{V \cup \{x' \}}$ are facets of $C$, or \item there exists a vertex $x''$ of $C$ that is vertex-reachable from both $x$ and $x'$.\end{enumerate*} In other words, vertices $x \ne x'$ are vertex-reachable if and only if there exists a finite sequence of vertices $\left\{ x_j \right\}_{j=1}^J$ of $C$ such that $x_1 = x$, $x_J = x'$, and for each $j = 1, \dots, J-1$ the vertices $x_j$ and $x_{j+1}$ lie on two neighbouring facets of $C$, but do not lie on the same facet of $C$. For example, in the top left panel of Figure~\ref{figure:motivation}, vertices $A$, $C$, $E$, and $G$ are mutually vertex-reachable, but $A$ is not vertex-reachable from $B$, $D$, $F$, or $H$.

% \begin{definition}
% 	Define a relation $S$ on the set of vertices of $C$. We say that $x_{i_1}$ and $x_{i_2}$ are \emph{related in $S$} if and only if either $x_{i_1} = x_{i_2}$, or if there exists a set $V$ of $d-1$ vertices of $C$ such that both $\conv{V \cup \{x_{i_1} \}}$ and $\conv{V \cup \{x_{i_2} \}}$ are facets of $C$. Let $S^\star$ be the transitive closure of $S$.
% \end{definition}

\begin{lemma}\label{lm:relations}
Let $H_1, H_2 \in \half(2)$ be two relevant halfspaces that cut off points $x_{j_1}$ and $x_{j_2}$, respectively, from $X$. If $x_{j_1}$ and $x_{j_2}$ are mutually vertex-reachable, then $H_1$ and $H_2$ are mutually reachable in $\half(2)$.
\end{lemma}

\begin{proof}
If $x_{j_1} = x_{j_2}$, then $H_1$ and $H_2$ are mutually reachable in $\half(2)$ by Lemma~\ref{lm:samepoint}. Otherwise, it suffices to consider the case $x_{j_1}$ and $x_{j_2}$ vertex-neighbouring with $x_{j_1} \ne x_{j_2}$, as the rest follows by induction. Let $x_{i_1}, \dots, x_{i_{d-1}}$ be the points from $X\setminus\left\{x_{j_1},x_{j_2}\right\}$ such that both convex hulls $F_1 = \conv{x_{j_1}, x_{i_1}, \dots, x_{i_{d-1}}}$ and $F_2 = \conv{x_{j_2}, x_{i_1}, \dots, x_{i_{d-1}}}$ are facets of $C$. Then, in the convex hull $C_1$ of $X \setminus \{x_{j_1} \}$ there exists a vertex $x_{h_1}$ such that $\conv{x_{h_1}, x_{i_1}, \dots, x_{i_{d-1}} } \ne F_2$ is a new facet of $C_1$. The corresponding halfspace $H_1^\prime \in \half(2)$ whose boundary is $\aff{x_{h_1}, x_{i_1}, \dots, x_{i_{d-1}}}$ that cuts off $x_{j_1}$ is relevant and reachable from $H_1$ in $\half(2)$ by Lemma~\ref{lm:samepoint}. Similarly there exists a relevant halfspace $H_2^\prime \in \half(2)$ with boundary $\aff{x_{h_2}, x_{i_1}, \dots, x_{i_{d-1}}}$, for some $x_{h_2} \in X \setminus \{ x_{j_2} \}$, that cuts off $x_{j_2}$ from $X$, and thus is reachable from $H_2$ in $\half(2)$. Since $H_1^\prime$ and $H_2^\prime$ share the ridge $\left\{ x_{i_1}, \dots, x_{i_{d-1}} \right\}$ of $X$, they are neighbouring, and thus mutually reachable in a single step in $\half(2)$. Since the relation of reachability of halfspaces is transitive and symmetric, also $H_1$ and $H_2$ are mutually reachable in $\half(2)$.
\end{proof}

\begin{lemma}\label{lm:vertices}
Let $F$ be a facet of $C$. Then, any vertex of $C$ is vertex-reachable from at least one of the vertices of $F$.
\end{lemma}

\begin{proof}
Let $x_{i_1}, \dots, x_{i_d}$ be the vertices of $F$. If $F^\prime$ is a neighbouring facet of $F$ (they share $d-1$ vertices), all vertices of $F^\prime$ are trivially vertex-reachable from some vertex of $F$. % Indeed, suppose that $F$ and $F^\prime$ share $x_{i_2}, \dots, x_{i_d}$. Then for $x_{i_2}, \dots, x_{i_d} \in F$ the claim is trivial because each vertex is vertex-related with itself. By definition of $S$, also $x_{i_1}$ is related in $S$ to the remaining vertex of $F^\prime$. 
Therefore, by induction, if $F''$ is a facet reachable from $F$ in the sense of Lemma~\ref{lm:connected}, then all vertices of $F''$ are vertex-reachable from some vertex of $F$. By Lemma~\ref{lm:connected}, all facets are reachable from $F$ and so the claim holds for all vertices of $C$.
\end{proof}

Putting together all our previous observations, we are now ready to prove Theorem~\ref{theorem:k2}. 

\smallskip
\noindent
\textbf{Main proof of Theorem~\ref{theorem:k2}.} For $k = 1$ the claim follows trivially from Lemma~\ref{lm:connected}. We therefore consider only $k=2$. The initial set of ridges $\mathcal Q$ of Algorithm~\aA{} includes all ridges formed by vertices of 
    \begin{itemize}
        \item a relevant halfspace in $\half(2)$ determined by $\aff{x_{i_1}, \dots, x_{i_d}}$ cutting off a point $x_j$, such that $x_{i_1}, \dots, x_{i_{d-1}}$ are vertices of the convex hull $C$ (step~\ref{A2}), and
        \item all halfspaces in $\half(2)$ that pass through $d-1$ points from $\{x_j, x_{i_1}, \dots, x_{i_{d-1}} \}$ (step~\ref{A3}).
    \end{itemize}  
This includes for each $l \in \{1, \dots, d-1\}$ a halfspace in $\half(2)$ that cuts off $x_{i_l}$.
	
The points $x_j, x_{i_1}, \dots, x_{i_{d-1}}$ are all vertices of a facet of $C$. If $H \in \half(2)$ is a relevant halfspace cutting off $x_h \in X$, by Lemma~\ref{lm:vertices} we know that $x_h$ is vertex-reachable from one of the points $x_j, x_{i_1}, \dots, x_{i_{d-1}}$, say $x_p$. By Lemma~\ref{lm:relations}, $H$ is reachable from the initial relevant halfspace that cuts off $x_p$. Since Algorithm~\aA{} finds all relevant halfspaces reachable from the initial set in step~\ref{S2}, it must find all relevant halfspaces from $\half(2)$.

\subsection{Proof of Theorem~\ref{next level}} \label{sec:next level}

To prove that~\ref{B1} does indeed guarantee exactness, we introduce additional notation. For a subset $A\subset X$ of $k-1$ points we write $\half(A)$ for the collection of all the halfspaces $H\in \half(k)$ that cut off $A$, i.e. that satisfy $X \setminus H = A$. Technically, $\half(A) \subset \half(k)$ depends also on $k$ and this should be emphasized in our notation; we shall not do this because $k$ is, in fact, implicitly present, as $\#A = k - 1$. We start with a simple but useful lemma.

\begin{lemma}  \label{simplex}
Let $S=\{a_1,\dots,a_{d+1}\} \subset \R^d$ be in general position. Denote by $p$ the projection of $a_{d+1}$ into the hyperplane $H = \aff{S\setminus\{a_{d+1}\}}$. Then there exists $j\in\{1,\dots,d\}$ such that $p$ and $a_j$ lie on the same side of the $(d-2)$-dimensional hyperplane inside $H$ determined by the $d-1$ points $S \setminus \{a_j,a_{d+1}\}$.
\end{lemma}

\begin{proof}
The set of $d$ points $S\setminus\{a_{d+1}\}$ lies in general position. Thus, it forms vertices of a $(d-1)$-dimensional simplex inside the hyperplane $H$. Consider now only the $(d-1)$-dimensional (affine) space $H$. Every vertex of the simplex lies inside a halfspace in $H$ bounded by the $(d-2)$-dimensional affine subspace (that is, a hyperplane in $H$) containing the remaining vertices. %can be separated from the remaining vertices by a hyperplane that contains the $d-1$ latter vertices \cite[Theorem~1.3.4]{Schneider2014}. The corresponding $d$ open halfspaces, which contain the remaining vertex,
These $d$ halfspaces cover $H$, and in particular, one of them contains $p$. We found a $(d-2)$-dimensional hyperplane in $H$ determined by $d-1$ points from $S\setminus\{a_{d+1}\}$ that bounds a halfspace containing $p$ and a single point from $S \setminus \{ a_{d+1} \}$, as desired.
\end{proof}

For the next lemma we need to consider also the metric structure of $\R^d$. Writing $\left\Vert x \right\Vert$ for the Euclidean norm of $x \in \R^d$, we define the distance of a point $x \in \R^d$ from a hyperplane $G \subset \R^d$ by $\dist(x;G) = \min_{y \in G} \left\Vert x - y \right\Vert$.

\begin{lemma}   \label{iterative step}
Let $H \in \half(k+1)$ be determined by points $a_1,\dots, a_d \in X$, and denote by $A = X \setminus H$ the set of $k$ points cut off from $X$ by $H$. Let $a \in A$ be any point of minimum distance to the hyperplane $\partial H$, that is, any $a \in A$ that satisfies $\dist(a;\partial H) = \min_{x \in A} \dist(x;\partial H)$. Then there exists an observational halfspace $H_1 \ne H$ that shares a ridge with $H$ in its boundary and either \begin{enumerate*}[label=(\roman*)] \item $H_1 \in \half(k)$, or \item $H_1$ cuts off from $X$ the same $k$ points as $H$ does, that is, $H_1 \in \half(A) \cap \half(k+1)$, and at the same time $\dist(a;\partial H_1) < \dist(a;\partial H)$. \end{enumerate*}
\end{lemma}

\begin{proof}
We start by applying Lemma~\ref{simplex} with $a_{d+1} = a$. Let $j \in \{1, \dots, d\}$ be the index from Lemma~\ref{simplex}. Denote by $G$ the unique halfspace that satisfies the following three conditions: \begin{enumerate*}[label=(\roman*)] \item the boundary of $G$ is orthogonal to the boundary of $H$ in the sense that the unit normals of $G$ and $H$ are orthogonal; \item $G$ contains the $d-1$ points $M=\{a_1,\dots,a_{j-1},a_{j+1},\dots,a_d\}$ in its boundary; and \item $a_j$ lies in the interior of $G$.\end{enumerate*} Then by Lemma~\ref{simplex} we have that also $a \in G$, because the projection $p$ of $a$ into $\partial H$ and $a_j$ lie on the same side of the projection of $\partial G$ into $\partial H$. Denote by $W$ a two-dimensional plane orthogonal to the $(d-2)$-dimensional affine hull of $M$. Write $o \in W$ for the projection of $M$ onto $W$ and for any other point $x\in\R^d$ denote by $x'\in W$ its projection into $W$. Since $M \subset \partial H$ is orthogonal to $W$, the projection of $H$ into $W$ is a halfplane $H'$ that does not contain $a'$ and its boundary is the line $\aff{o,a'_j}$. Likewise, because $M \subset \partial G$, also the projection of $G$ is a halfplane $G'$ in $W$ that contains $o$ on its boundary and has inner normal $a'_j-o$. Because $a \in G$, we obtain that $a'\in G'$. For any data point $b \in X$ denote by $\theta(b)$ the value of the smaller angle between the lines $\aff{o,b'}$ and $\aff{o,a'_j}$ inside the plane $W$. This whole setup is visualised in Figure~\ref{fig:iterative}.

\begin{figure}[htpb]
    \centering
    \includegraphics[width=.43\textwidth]{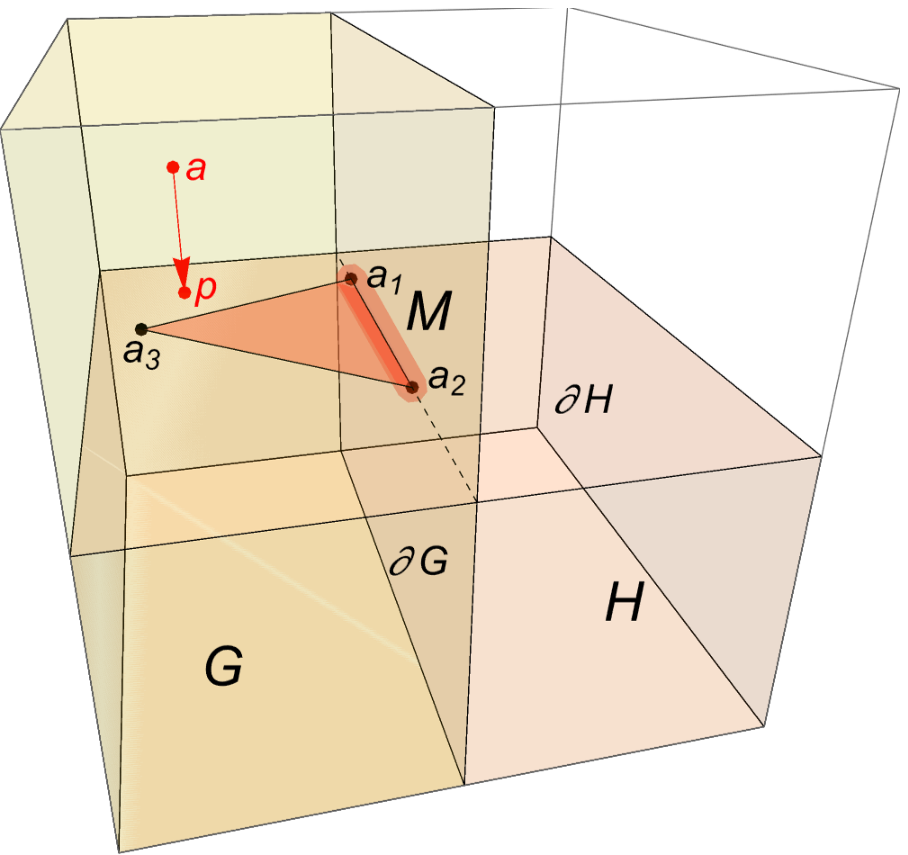} \quad
    \includegraphics[width=.43\textwidth]{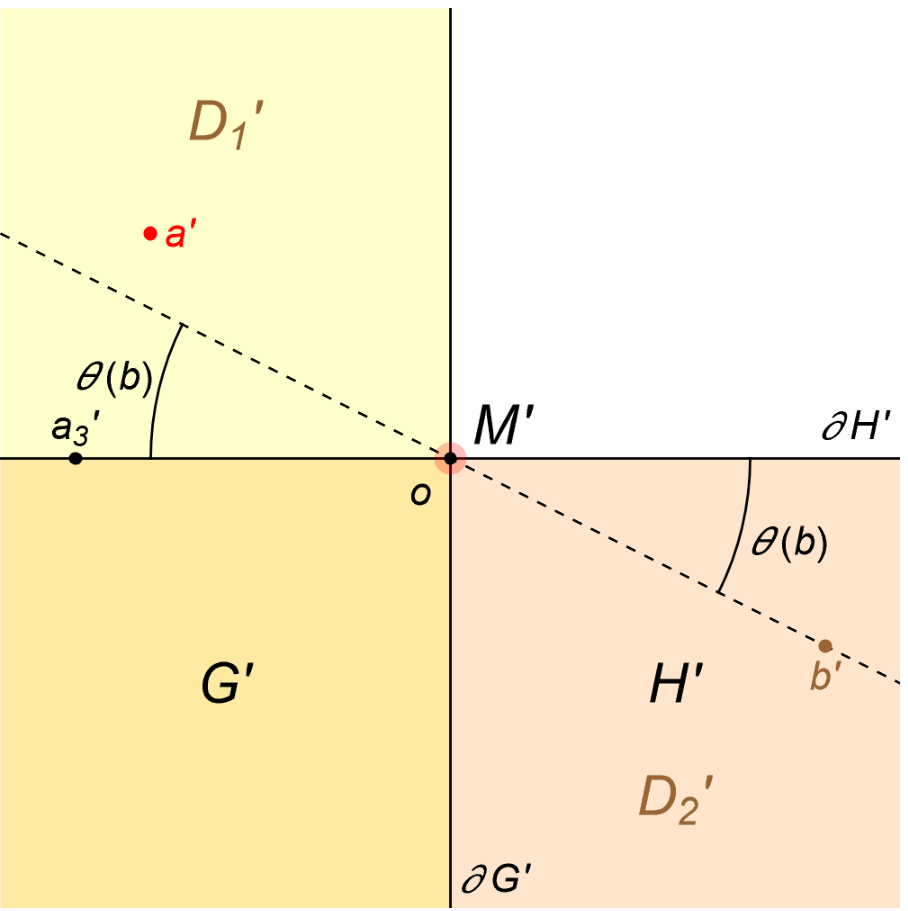} 
    \caption{Visualisation of the proof of Lemma~\ref{iterative step}. In the left panel we see the setup in dimension $d=3$. The boundary plane $\partial G$ of the halfspace $G$ is the vertical plane passing through the ridge $M$. Orthogonal to $\partial G$ we have the boundary plane $\partial H$, which also passes through $M$ and delimits the halfspace $H$ in the lower part of the figure. The index from Lemma~\ref{simplex} is $j=3$; the plane $W$ is any plane with a normal vector $a_1 - a_2$. In the panel on the right hand side we have our setup projected into $W$. The ridge $M$ projects into the singleton $M'=\{o\}$. The regions $D_1$ and $D_2$ are projected to the quadrants given by $D_1'$ and $D_2'$. We search for a point $b \in D \cap X$ that minimizes the value of the angle $\theta(b)$ between the line $\partial H'$ and the line joining $b' \in D_1' \cup D_2'$ with $o$.}
    \label{fig:iterative}
\end{figure}

Denote 
    %\[  
    %\begin{aligned}
    $D_1 = G \setminus H$, 
    %\quad \mbox{ and } \quad 
    and $D_2 = H \setminus G$,
    %\end{aligned}
    %\]
and set $D = D_1 \cup D_2$. Let $c\in D \cap X$ be the point that satisfies
    \begin{equation*}
    %\label{theta}
    \theta(c)=\min\left\{\theta(b)\colon b \in D \cap X\right\}.    
    \end{equation*}
Such a point is necessarily unique, as $M$, which projects to $o \in W$, already contains $d-1$ points from $X$, and if there were two different points $c_1, c_2$ from $X \setminus M$ with the same angle $\theta(c)$, there would necessarily exist a hyperplane in $\R^d$ passing through $M$ and both $c_1$ and $c_2$ in $\R^d$, which is impossible due to the assumption of $X$ being in general position.

Denote by $H_1$ the halfspace in $\R^d$ with boundary determined by $M\cup\{c\}$ that does not contain $a$ in its interior. Note that $H_1$ is actually a halfspace obtained by rotating $H$ around $M$ in $\R^d$ in the direction that ``keeps" $a_j$ in $H_1$, until its boundary $\partial H_1$ hits the first new point from $X$. Also, $H_1$ obviously shares the ridge $M$ with $H$.

 We distinguish two cases. 
\begin{enumerate*}[label=(\roman*)]
    \item $c\in D_1$: In this case, we know that $c$ is one of the $k$ points in $A$. Consequently, $H_1$ cuts off exactly $k-1$ points from $X$, that is, $H_1 \in \half(k)$, and $H_1 \in \half(A \setminus \{c\})$. 
    \item $c\in D_2$: This situation implies that $H_1$ cuts off the same $k$ points as $H$ does, i.e. $H_1 \in \half(A)$. What remains to be proved is that $\dist(a;\partial H) > \dist(a;\partial H_1)$. To see this, note that because the plane $W$ is parallel to normal vectors of both $\partial H$ and $\partial H_1$, we have that for a given point $x\in\R^d$ the distances $\dist(x;\partial H)$ and $\dist(x;\partial H_1)$ are equal to the distances of the projection of the point $x'\in W$ from the lines $\aff{o,a'_i}$ and $\aff{o,c'}$, respectively, in the plane $W$. Therefore, it is enough to prove that $\dist(a';\aff{o,a'_i}) > \dist(a';\aff{o,c'})$, which follows directly from $\theta(a') > \theta(c)$ and from the fact that $a,c\in D$. 
\end{enumerate*}
\end{proof}

Everything is now ready for the main proof of Theorem~\ref{next level}.

\smallskip
\noindent
\textbf{Main proof of Theorem~\ref{next level}.} Denote by $A = X \setminus H$ the set of $k$ points cut off from $X$ by $H$. We apply Lemma~\ref{iterative step} to the halfspace $H$ and obtain a halfspace $H_1$. If $H_1$ cuts off $k-1$ points from $X$, then we set $H'=H$ and $\widetilde{H}=H_1$. Otherwise, we know that $H_1 \in \half(k+1) \cap \half(A)$ and for $a \in A$ from the statement of Lemma~\ref{iterative step} we have $\min_{x \in A} \dist(x;\partial H) = \dist(a;\partial H) > \dist(a;\partial H_1) \geq \min_{x \in A} \dist(x;\partial H_1)$. We apply again Lemma~\ref{iterative step} to the halfspace $H_1 \in \half(k+1)$ and obtain another halfspace $H_2$. Note that $H_2$ is reachable from $H$ by our construction. Again, if $H_2$ cuts off $k-1$ points, then we set $H'=H_1$ and $\widetilde{H}=H_2$; otherwise, we continue and apply Lemma~\ref{iterative step} iteratively. We obtain a sequence $H_1, H_2, H_3, \dots$. Because in each step we obtain a halfspace with a strictly smaller distance from the closest point of the set $A$, it is not possible for any two halfspaces from the sequence to coincide. Because there are only finitely many observational halfspaces, there must exist an integer $m \geq 2$ such that $H_m \in \half(k)$. Taking $m$ to be the lowest such integer, we set $H'=H_{m-1}$ and $\widetilde{H}=H_m$. Note that $H_{m-1}$ is a relevant halfspace that cuts off $k$ points from $X$ and also is reachable from $H$. The proof is concluded.

\subsection*{Supplementary material.} % The online Supplementary Material accompanying this paper contains:
    \begin{itemize}
    \item An updated \proglang{R} package \pkg{TukeyRegion}, version~0.1.5.5 where the novel exact Algorithm~\aB{} is implemented.
    \item A \proglang{pdf} file with an additional Algorithm~\aD{} motivated by an extension of Algorithm~\aA{} with $k=2$. Using the dual graph, we present a dataset where also this possible simplification of Algorithms~\aB{} and~\aCmb{} fails to recover the central region. Further, we propose to use the dual graph for heuristic assessment of the quality of approximation using non-exact algorithms like Algorithms~\aA{},~\aC{} or~\aD{}. % We describe the new functionalities implemented in the \proglang{R} package~\pkg{TukeyRegion} that allow us to (partially) reconstruct the dual graph of $X$.
    \item A \proglang{Mathematica} notebook with functions for computing the dual graph of $X$, containing also interactive visualisations of all the examples provided in this paper.
    %\item Complete \proglang{R} source codes for the simulation studies performed in Section~\ref{sec:simulation}. 
    \end{itemize}

\subsection*{Acknowledgement}
P.~Laketa was supported by the OP RDE project ``International mobility of research, technical and administrative staff at the Charles University", grant number CZ.02.2.69/0.0/0.0/18\_053/0016976. The work of S.~Nagy was supported by Czech Science Foundation (EXPRO project n. 19-28231X).

%\spacingset{1.2}
%\bibliographystyle{apalike}
%\bibliography{allpapers,nagy}

\def\cprime{$'$} \def\polhk#1{\setbox0=\hbox{#1}{\ooalign{\hidewidth
  \lower1.5ex\hbox{`}\hidewidth\crcr\unhbox0}}}

\end{document}